\documentclass[11pt,oneside,fleqn]{article}

\usepackage[utf8]{inputenc}
\usepackage{hyperref}
\usepackage{amsmath,amssymb,amsthm}
\usepackage{graphicx}
\usepackage{color}
\usepackage{multirow}
\usepackage[mathscr]{eucal}

\allowdisplaybreaks

\hypersetup{colorlinks, linkcolor=blue, citecolor=blue, urlcolor=blue}
\renewcommand{\arraystretch}{1.2}

\newcommand{\dd}[2]{\frac{\mathrm{d} #1}{\mathrm{d} #2}}

\newcommand{\p}{\partial}

\newcommand{\const}{\mathop{\rm const}\nolimits}
\newcommand{\rank}{\mathop{\rm rank}\nolimits}
\newcommand{\todo}[1][\null]{\ensuremath{\clubsuit}}

\newcommand{\checked}[1][\null]{\ensuremath{\boldsymbol{\surd}}}

\newtheorem{theorem}{Theorem}
\newtheorem{lemma}{Lemma}
\newtheorem{corollary}{Corollary}
\newtheorem{proposition}{Proposition}
{\theoremstyle{definition}
\newtheorem{definition}{Definition}
\newtheorem{example}{Example}
\newtheorem*{example*}{Example}
\newtheorem{remark}{Remark}
\newtheorem*{remark*}{Remark}
}

\newcommand{\lsemioplus}{\mathbin{\mbox{$\lefteqn{\hspace{.77ex}\rule{.4pt}{1.2ex}}{\in}$}}}

\newcommand{\DD}{\mathrm{D}}

\newcommand{\vv}{\mathbf{v}}
\newcommand{\ww}{\mathbf{w}}

\newcommand{\ve}{\varepsilon}
\newcommand{\Ad}{\mathrm{Ad}}

\newcommand{\DDD}{\mathcal{D}}
\newcommand{\GG}{\mathcal{G}}

\begin{document}

\par\noindent {\LARGE\bf
Enhanced preliminary group classification\\  of a class of generalized diffusion equations
\par}
{\vspace{4mm}\par\noindent {\bf Elsa Dos Santos Cardoso--Bihlo$^\dag$, Alexander Bihlo$^\dag$ and Roman O. Popovych$^\dag\, ^\ddag$
} \par\vspace{2mm}\par}

{\vspace{2mm}\par\noindent {\it
$^{\dag}$~Faculty of Mathematics, University of Vienna, Nordbergstra{\ss}e 15, A-1090 Vienna, Austria\\
}}
{\noindent \vspace{2mm}{\it
$\phantom{^\dag}$~\textup{E-mail}: elsa.cardoso@univie.ac.at, alexander.bihlo@univie.ac.at
}\par}

{\vspace{2mm}\par\noindent {\it
$^\ddag$~Institute of Mathematics of NAS of Ukraine, 3 Tereshchenkivska Str., 01601 Kyiv, Ukraine\\
}}
{\noindent \vspace{2mm}{\it
$\phantom{^\dag}$~\textup{E-mail}: rop@imath.kiev.ua
}\par}

\vspace{2mm}\par\noindent\hspace*{8mm}\parbox{140mm}{\small
The method of preliminary group classification is rigorously defined, enhanced and related to the theory of group classification of differential equations. Typical weaknesses in papers on this method are discussed and strategies to overcome them are presented. The preliminary group classification of the class of generalized diffusion equations of the form $u_t=f(x,u)u_x^2+g(x,u)u_{xx}$ is carried out. This includes a justification for applying this method to the given class, the simultaneous computation of the equivalence algebra and equivalence (pseudo)group, as well as the classification of inequivalent appropriate subalgebras of the whole infinite-dimensional equivalence algebra. The extensions of the kernel algebra, which are induced by such subalgebras, are exhaustively described. These results improve those recently published in \textit{Commun.\ Nonlinear Sci.\ Numer.\ Simul.}

}\par\vspace{2mm}

\section{Introduction}

Group classification of differential equations is an efficient tool for investigating symmetry properties of classes of differential equations. These are differential equations that include arbitrary constants or functions of the independent and dependent variables as well as of derivatives of the dependent variables up to a certain order. It is known for a long time that depending on the value of these arbitrary elements the resulting differential equations from the given class can have different Lie invariance groups. The first examples of group classification were presented by Sophus Lie for the class of second order linear partial differential equations~\cite{lie81Ay} and the class of second order ordinary differential equations~\cite{lie91Ay}. Later, Ovsiannikov began the rigorous development of the theory of group classification~\cite{ovsi82Ay}. In short, the solution of the group classification problem consists in finding the kernel of Lie invariance groups (i.e.\ those Lie symmetries that are admitted for all values of the arbitrary elements) and all inequivalent extensions of Lie invariance groups with respect to the kernel group. 
The equivalence involved means the similarity of equations up to transformations from a certain equivalence group (e.g.\ usual, generalized or conditional equivalence), see~\cite{popo10Ay} for more detailed information.

For classes of differential equations being of simple structure (e.g., ones parameterized only by constants or functions of the same single argument), the corresponding group classification problems can be completely solved via compatibility analysis and explicit integration of the determining equations for Lie symmetries depending on values of the arbitrary elements and up to the equivalence chosen. Complete group classification can also be carried out for classes of differential equations possessing the normalization property. The algebraic method of classification effectively works for such classes. See the next section and also~\cite{popo10Ay,vane09Ay} for a more comprehensive review on different methods of group classifications.

In the situation where the class depends in a more complicated way on its arbitrary elements, it may happen that both the determining equations are too difficult to be directly solved and the application of the algebraic method does not give the exhaustive solution. In this case, however, at least a partial solution of the group classification problem, known as \emph{preliminary group classification}, is possible. The basic idea of preliminary group classification is to study only those extensions of the kernel group that are induced by the transformations from the corresponding equivalence group. The problem of finding inequivalent cases of such Lie symmetry extensions then reduces to the classification of inequivalent subgroups (resp.\ algebras) of the equivalence group (resp.\ algebra). This approach was first described in~\cite{akha91Ay} and became prominent due to the paper~\cite{ibra91Ay}. 

Despite the approach of preliminary group classification is rather common, it is not well developed up to now. The basic mechanisms were formulated in~\cite{ibra91Ay} as two propositions for the specific class of nonlinear wave equations of the form $v_{tt}=f(x,v_x)v_{xx}+g(x,v_x)$ and were later adopted in other papers for respective classes of equations. In the present paper, we state a stronger version of these propositions for general classes of differential equations. Another weakness commonly observed is that, when the equivalence algebra~$\mathfrak g^\sim$ of the class of equations under consideration is infinite dimensional, only Lie symmetry extensions induced by subalgebras of a finite-dimensional subalgebra~$\mathfrak g^\sim_0$ of~$\mathfrak g^\sim$ are investigated, without giving any sound justification for the choice of~$\mathfrak g^\sim_0$. In fact, this restriction is needless as it is possible to classify subalgebras of infinite-dimensional algebras in much the same way as subalgebras of finite-dimensional algebras
\cite{basa01Ay,bihl10Ay,bihl09Ay,bihl10Dy,fush94Ay,lahn06Ay,maga93Ay,popo10Cy,popo10Ay,zhda99Ay}. 
It can even be simpler to classify low-dimensional subalgebras of the whole infinite-dimensional equivalence algebra~$\mathfrak g^\sim$ as the adjoint action related to~$\mathfrak g^\sim$ is more powerful and allows for greater simplification than the adjoint action corresponding to the finite-dimensional subalgebra~$\mathfrak g^\sim_0$. One more common weakness in papers on the subject is that usually only extensions induced by one-dimensional subalgebras of  equivalence algebras are studied. Moreover, these one-dimensional subalgebras (of a finite-dimensional subalgebra~$\mathfrak g^\sim_0$ of~$\mathfrak g^\sim$) are classified only with respect to the restricted equivalence relation which is generated by the adjoint representation of~$\mathfrak g^\sim_0$. This leads to an overly large number of inequivalent subalgebras compared to the list of one-dimensional subalgebras that would be obtainable if the classification was done using the adjoint representation of the entire~$\mathfrak g^\sim$.

In the present paper, we comprehensively carry out preliminary group classification for the class of $(1+1)$-dimensional second order quasilinear evolution equations of the general form
\begin{gather}\label{eq:GenDiffEqs}
 \Delta = u_t - f(x,u)u_x^2 - g(x,u) u_{xx} = 0,
\end{gather}
where~$f$ and~$g$ are arbitrary smooth functions of~$x$ and~$u$, and $g\ne0$.
The class~\eqref{eq:GenDiffEqs} was considered in the recent paper~\cite{nadj10Ay} but results obtained therein are not correct.
It is reviewed above that there are a number of typical weaknesses in papers on preliminary group classification, and results of~\cite{nadj10Ay} properly illustrate these weaknesses, cf.\ the first paragraphs of Sections~\ref{sec:EquivalenceAlgebra} and~\ref{sec:ClassificationOfSubalgebras} and Remark~\ref{rem:OnNadj10Ay}. This is why we aim to accurately present the revised preliminary group classification of the class~\eqref{eq:GenDiffEqs} and to give all calculations in considerable detail.

The class~\eqref{eq:GenDiffEqs} was considered in~\cite{nadj10Ay} as a class of generalized Burgers equations as it includes the \emph{potential Burgers equation} as a particular element for the choice~$f=g=1$. This class also contains $(1+1)$-dimensional linear evolution equations, which correspond to the values~$f=0$ and~$g$ not depending on~$u$. As a prominent example for a linear differential equation, one can recover the linear heat equation by choosing $f=0$ and $g=1$. 
An important subclass of the class~\eqref{eq:GenDiffEqs} is the class of $(1+1)$-dimensional nonlinear diffusion equations of the general form $u_t=(F(u)u_x)_x$, where $F\ne0$. It is singled from the class~\eqref{eq:GenDiffEqs} by the constraints $g_x=0$ and $f=g_u$. 
Moreover, any equation of the form~\eqref{eq:GenDiffEqs} with $f_x=g_x=0$ is reduced to a diffusion equation by a simple point transformation acting only on the dependent variable~$u$.  
The solution of the group classification problem for this class by Ovsiannikov \cite{ovsi59Ay} (see also \cite{akha91Ay,ovsi82Ay})
gave rise to the development of modern group analysis. 

The class~\eqref{eq:GenDiffEqs} is included in the wider class of equations $u_t = F(t,x,u,u_x)u_{xx}+G(t,x,u,u_x)$, for which the complete group classification was carried out in~\cite{basa01Ay} by a method similar to that applied in the present paper. The fact that this class is normalized (cf.\ Section~\ref{sec:EnhancedMethodOfPreliminaryGroupClassification}) played a crucial role in the entire consideration in~\cite{basa01Ay}. However, as this class is essentially wider than the class~\eqref{eq:GenDiffEqs}, the corresponding equivalence algebras are rather different. This is why the results of~\cite{basa01Ay} cannot be directly used for deriving the group classification of the class~\eqref{eq:GenDiffEqs}.

It should also be stressed that it is not natural to exclude linear differential equations from the present consideration. In fact, there are equations in the class~\eqref{eq:GenDiffEqs} which are linearized by point transformations from the equivalence group $G^\sim$ of this class. The most prominent example of such a transformation in the above class is the transformation of the potential Burgers equation to the linear heat equation by means of the point transformation $\tilde u=e^u$~\cite[p.~122]{olve86Ay}. That is, $u$ is a solution of the potential Burgers equation whenever $\tilde u$ is a solution of the linear heat equation. In the course of preliminary group classification of the class~\eqref{eq:GenDiffEqs} we encounter other examples of linearizable equations. Furthermore, the equivalence algebra~$\mathfrak g^\sim_0$  of the subclass of~\eqref{eq:GenDiffEqs}, which is compliment to the subclass of linear equations and, therefore, singled out by the constraint $f^2+g_u\!{}^2\ne0$, is much narrower than the equivalence algebra~$\mathfrak g^\sim$ of the entire class~\eqref{eq:GenDiffEqs}. 
More precisely, the algebra~$\mathfrak g^\sim_0$ is singled out as a subalgebra of~$\mathfrak g^\sim$ by the constraint~$h_{uu}=0$, cf.\ Theorem~\ref{thm:EquivalenceAlgebra}.

The further organization of this paper is the following. The subsequent Section~\ref{sec:EnhancedMethodOfPreliminaryGroupClassification} discusses the theory of preliminary group classification. We generalize and extend assertions presented in~\cite{ibra91Ay} and formulate them rigorously using the modern language of group analysis. In Section~\ref{sec:DeterminingEquationsOfLieSymmetries} we derive the determining equations for Lie point symmetries of equations from the class~\eqref{eq:GenDiffEqs} and find the corresponding kernel of Lie invariance algebras. The equivalence algebra~$\mathfrak g^\sim$ and the equivalence group~$G^\sim$ of the class~\eqref{eq:GenDiffEqs} is computed in Sections~\ref{sec:EquivalenceAlgebra} and~\ref{sec:EquivalenceGroup}, respectively.  Throughout the paper, by the equivalence group we mean the Lie pseudo-group of point equivalence transformations (i.e., local equivalence diffeomorphisms), cf.\ \cite{olve09By} and references therein for theory of pseudo-groups. The reason for carrying out preliminary group classification is elucidated. In Section~\ref{sec:ClassificationOfSubalgebras}, we classify inequivalent one- and two-dimensional subalgebras of the essential subalgebra of~$\mathfrak g^\sim$. The corresponding inequivalent cases of symmetries extensions of the kernel algebra are presented in Section~\ref{sec:PreliminaryGroupClassification} and supplemented with three- and four-dimensional extensions via the classification of all appropriate subalgebras of~$\mathfrak g^\sim$. The paper concludes with a short summary and further comments in Section~\ref{sec:Conclusion}.

\section{Enhanced method of preliminary group classification}\label{sec:EnhancedMethodOfPreliminaryGroupClassification}

By now, the method of preliminary group classification was neither explained for general classes of differential equations nor properly related to the general group classification problem. This should be done first in this section before we study the preliminary group classification of~\eqref{eq:GenDiffEqs}. For this aim, we need a few notions of the theory of group classifications, which can be found in the recent paper~\cite{popo10Ay}.

The most essential notion concerns the formal definition of \emph{classes of differential equations}. In general, a class (of systems) of differential equations is given by a system of $l$ differential equations of the form $L(x,u_{(p)},\theta(x,u_{(p)}))=0$ in $m$ dependent variables $u=(u^1,\dots,u^m)$ and $n$ independent variables $x=(x_1,\dots,x_n)$, where $u_{(p)}$ denotes the set of $u$'s and all their derivatives up to order $p$. The differential functions $\theta(x,u_{(p)})=(\theta^1(x,u_{(p)}),\dots,\theta^k(x,u_{(p)}))$ denote a tuple of $k$ arbitrary elements that parameterize the given class of differential equations. The tuple $\theta$ is usually constrained to satisfy a system~$\mathcal S$ of auxiliary conditions, $S(x,u_{(p)},\theta_{(q)}(x,u_{(p)}))=0$, in which $x$ and $u_{(p)}$ are regarded as independent variables. The set of solutions of this auxiliary system will also be denoted by~$\mathcal S$. In addition, this set can be further constrained by satisfying one or more nonvanishing conditions $\Sigma(x,u_{(p)},\theta_{(q)}(x,u_{(p)}))\ne0$. Putting together all these notions, we denote the class of differential equations with the arbitrary element running through the set $\mathcal S$ by $\mathcal L|_{\mathcal S}$. The single elements of this class are denoted by $\mathcal L_\theta$, respectively.

Specifically, for the class~\eqref{eq:GenDiffEqs} we have $\theta=(f,g)$, and the arbitrary elements~$f$ and~$g$ depend only on~$x$ and~$u$. 
Therefore, the associated auxiliary system~$\mathcal S$ is formed by the equations
\begin{gather*}
f_t=f_{u_t}=f_{u_x}=f_{u_{tt}}=f_{u_{tx}}=f_{u_{xx}}=0,\\
g_t=g_{u_t}=g_{u_x}=g_{u_{tt}}=g_{u_{tx}}=g_{u_{xx}}=0.
\end{gather*}
The auxiliary conditions $f_t=0$ and $g_t=0$ play a special role. All the other auxiliary conditions can be taken into account implicitly. 
The nonvanishing condition associated with the class~\eqref{eq:GenDiffEqs} is $g\ne0$, i.e., we have $\Sigma=g$. 
The condition $g\ne0$ should be explicitly included in the definition of the class~\eqref{eq:GenDiffEqs} since equations of the same form with $g=0$  
are of another (first) order, possess completely different transformational properties and are not related to equations with $g\ne0$ by point or other reasonable transformations.

Having properly defined classes of differential equations, it remains to introduce the notion of admissible transformations and normalized classes of differential equations in order to explain the general strategy of (preliminary) group classification.

\begin{definition}
The set of admissible transformations in the class $\mathcal L|_{\mathcal S}$ is given by $\mathrm T(\mathcal L|_{\mathcal S})=\{(\theta,\tilde \theta, \varphi)\ |\ \theta,\tilde \theta\in\mathcal S, \varphi \in\mathrm T(\theta,\tilde \theta)\}$, where $\mathrm T(\theta,\tilde \theta)$ denotes the set of point transformations that map the system $\mathcal L_\theta$ to the system $\mathcal L_{\tilde \theta}$.
\end{definition}

The set of admissible transformations can be used for many issues related to the problem of group classification. It can be considered as an extension or generalization of the equivalence group of a class of differential equations. Indeed, the usual equivalence group~$G^\sim$ of a class $\mathcal L|_{\mathcal S}$ is naturally embedded in the set of admissible transformations. In particular, it is given by the admissible transformations $(\theta,\Phi\theta,\Phi|_{(x,u)})$, where $\Phi$ is an equivalence transformation, i.e.\ $\forall \theta\in\mathcal S\colon \Phi\theta\in\mathcal S$. In this last tuple, $\Phi|_{(x,u)}\in\mathrm{T}(\theta,\Phi\theta)$ denotes the projection of $\Phi$ to the space of variables $x,u$.
The maximal point symmetry group $G_\theta$ of the system $\mathcal L_\theta$ coincides with the set of admissible transformations from $\mathcal L_\theta$ to itself, i.e., $G_\theta=\mathrm T(\theta,\theta)$.

Important properties of classes of differential equations, relevant for the problem of group classification, are given by different kinds of normalization 
with respect to point (resp.\ contact) transformations~\cite{popo06Ay,popo10Ay}.  

\begin{definition}
The class of differential equations $\mathcal L|_{\mathcal S}$ is \emph{normalized} in the usual sense if any admissible transformation is induced by a transformation of the (usual) equivalence group, i.e.\ $\forall (\theta,\tilde \theta,\varphi)$, $\exists \Phi\in G^\sim\colon \tilde \theta = \Phi\theta$ and $\varphi=\Phi|_{(x,u)}$.
\end{definition}

Denote by $\mathfrak g_\theta$ the maximal Lie invariance algebra of the equation $\mathcal L_\theta$.
Using the above notations it is possible to obtain the general picture of the group classification problem. The first step in order to carry out group classification is the determination of the \emph{kernel}~$\mathfrak g^\cap$ (i.e., intersection) of maximal Lie invariance algebras of systems from the class~$\mathcal L|_{\mathcal S}$. The kernel is found by deriving the determining equations of Lie symmetries and splitting with respect to both derivatives with respect to $u$ and the arbitrary elements~$\theta$. This gives those part of the maximal Lie invariance algebra~$\mathfrak g_\theta$ that is admitted for any value of~$\theta$. The subsequent step consists of determining the equivalence group~$G^\sim$ (resp.\ the equivalence algebra~$\mathfrak g^\sim$) of the class~$\mathcal L|_{\mathcal S}$. The equivalence group~$G^\sim$ is needed since it generates a natural equivalence relation on cases of symmetry extension of the kernel and hence they should be studied up to this equivalence. The final task is to describe all inequivalent cases of symmetry extension, i.e., values of~$\theta$ for which $\mathfrak g_\theta\ne\mathfrak g^\cap$.

For the implementation of the above classification program, several special techniques have been developed. They either lead to the \emph{complete group classification} or to a \emph{preliminary group classification} of the given class.

Complete group classification is often possible for normalized classes of differential equations. For such classes, symmetry extensions of the kernel algebra can only be induced by vector fields from the corresponding equivalence algebra. This reduces the group classification problem to the algebraic problem of classifying inequivalent subalgebras of the equivalence algebra. This is why we refer to this method as the algebraic method. Results on complete group classification of various classes of differential equations can be found, e.g.\ in \cite{akha91Ay,basa01Ay,lahn06Ay,popo10Cy,popo10Ay,popo08Ay,popo10By,zhda99Ay}. 
The equations studied in these papers all possess the normalization property.

Another method leading to the complete solution of the group classification problem consists of a compatibility analysis and direct integration of the determining equations of Lie symmetries of the given class~\cite{akha91Ay,ibra94Ay,ivan10Ay,niki05Ay,niki06Ay,niki07Ay,niki01Ay,ovsi82Ay,vane09Ay}. It was indicated in the introduction that it is often only for rather simple classes that this method works.

\emph{Complete preliminary group classification} employs essentially the same techniques that are used for complete group classification within the framework of the algebraic method. The main difference is that the underlying class does not possess the normalization property. This implies the existence of extensions of the kernel algebra that are not induced by subalgebras of the equivalence algebra. In turn, for normalized classes of differential equations the results of complete preliminary group classification and complete group classification coincide~\cite{popo06Ay,popo10Ay}.

In most papers on preliminary group classification only a partial solution of the corresponding problems is achieved since usually not the whole equivalence algebra is used for an investigation of cases of symmetry extensions. This is why we refer to this method as the method of \emph{partial preliminary group classification}. It is the most incomplete and heuristic method of group classification, as there are often no obvious criteria which subalgebras of the equivalence algebra to single out for an investigations of symmetry extensions of the kernel algebra. Results on partial preliminary group classification are presented, e.g., in~\cite{akha91Ay,ibra91Ay,nadj10Ay,song09Ay}.

On the side of complete group classification, the theoretical background was already settled~\cite{lisl92Ay,ovsi82Ay} and extended~\cite{popo06Ay,popo10Ay}. It remains to detail the framework of preliminary group classification. In its essence, it rests on the following two propositions, which were first formulated without proof in~\cite{ibra91Ay} for the class of equations investigated. We present an enhanced version of these propositions for general classes of differential equations.

\begin{proposition}\label{pro:OnPreliminaryGroupClassification1}
Let $\mathfrak a$ be a subalgebra of the equivalence algebra~$\mathfrak g^\sim$ of the class $\mathcal L|_{\mathcal S}$, $\mathfrak a\subset\mathfrak g^\sim$, and let $\theta^0(x,u_{(r)})\in\mathcal S$ be a value of the tuple of arbitrary elements~$\theta$ for which the algebraic equation $\theta=\theta^0(x,u_{(r)})$ is invariant with respect to~$\mathfrak a$. Then the differential equation $\mathcal L_{\theta^0}$ is invariant with respect to the projection of $\mathfrak a$ to the space of variables $(x,u)$.
\end{proposition}

\begin{proof}
Choose an arbitrary operator~$Q$ from $\mathfrak a$ and consider the one-parameter group~$G_1$ generated by this operator. As the equation $\theta=\theta^0(x,u_{(r)})$ is invariant with respect to~$G_1$, any transformation~$\mathcal T$ from~$G_1$ maps the corresponding equation $\mathcal L_{\theta^0}$ from the class $\mathcal L|_{\mathcal S}$ to itself. This means that the projection ${\rm P}\mathcal T$ of~$\mathcal T$ to the space of variables $(x,u)$ is a point symmetry of~$\mathcal L_{\theta^0}$. Therefore, the projection ${\rm P}G_1$ of~$G_1$ is a point symmetry group of ~$\mathcal L_{\theta^0}$ and its generator, which is the projection of the operator~$Q$, belongs to the Lie invariance algebra of~$\mathcal L_{\theta^0}$.
\end{proof}

\begin{proposition}\label{pro:OnPreliminaryGroupClassification2}\looseness=-1
Let $\mathcal S_i$ be the subset of~$\mathcal S$ that consists of all arbitrary elements for which the corresponding algebraic equations are invariant with respect to the same subalgebra of the equivalence algebra~$\mathfrak g^\sim$ and
let $\mathfrak a_i$ be the maximal subalgebra of~$\mathfrak g^\sim$ for which $\mathcal S_i$ satisfies this property, $i=1,2$.
Then the subalgebras~$\mathfrak a_1$ and~$\mathfrak a_2$ are equivalent with respect to the adjoint action of~$G^\sim$ if and only if the subsets~$\mathcal S_1$ and~$\mathcal S_2$ are mapped to each other by transformations from~$G^\sim$.
\end{proposition}

\begin{proof}
Assume that $\mathfrak a_2 = \mathcal T_* \mathfrak a_1$, where $\mathcal T\in G^\sim$ and $\mathcal T_*$ denotes the associated push-forward of vector fields. For $\theta^0\in\mathcal S_1$ the algebraic equation $\theta=\theta^0$ is invariant with respect to $\mathfrak a_1$. Since $\mathcal T$ is an equivalence transformation, we also have that $\mathcal T \theta^0\in \mathcal S$. By supposition, the equation $\tilde \theta=\mathcal T \theta^0$ is invariant with respect to $\mathcal T_*\mathfrak a_1=\mathfrak a_2$. This implies that $\mathcal T \theta^0\in \mathcal S_2$ from which it can be concluded that $\mathcal T \mathcal S_1 \subset \mathcal S_2$. Similarly, for $\tilde \theta^0\in\mathcal S_2$, the algebraic equation $\tilde \theta= \tilde \theta^0$ is invariant with respect to $\mathfrak a_2$ and $\mathcal T^{-1}\tilde \theta^0\in\mathcal S$. As $\mathcal T^{-1}_*\mathfrak a_2 = \mathfrak a_1$, the algebraic equation $\theta=\mathcal T^{-1} \tilde \theta^0$ is invariant with respect to $\mathfrak a_1$, which implies that $\mathcal T^{-1}\tilde \theta^0\in S_1$. From this last condition we obtain $\mathcal T \mathcal S_1\supset \mathcal S_2$. It therefore can be concluded that there exists a bijection between $\mathcal S_1$ and $\mathcal S_2$, generated by a transformation from~$G^\sim$.

Conversely, suppose that $\mathcal S_2=\mathcal T \mathcal S_1$ for $\mathcal T\in G^\sim$. If $\theta=\theta^0$ is invariant with respect to $\mathfrak a_1$ then $\tilde \theta=\mathcal T \theta^0$ is invariant with respect to $\mathcal T_*\mathfrak a_1$. As $\theta^0$ is arbitrary, this implies that $\mathcal T_* \mathfrak a_1\subset \mathfrak a_2$. In a similar manner as in the previous paragraph, we can show that $\mathcal T_*^{-1}\mathfrak a_2\subset\mathfrak a_1$ using the inverse transformation of $\mathcal T$. Then we have $\mathcal T_*\mathcal T_*^{-1}\mathfrak a_2\subset \mathcal T_*\mathfrak a_1$ and thus $\mathfrak a_2 \subset \mathcal T_*\mathfrak a_1$. This is why $\mathfrak a_2 = \mathcal T_* \mathfrak a_1$ must hold, which completes the proof of the proposition.
\end{proof}

Roughly speaking, the first proposition defines the method of how to construct cases of symmetry extensions if the equivalence algebra of the class of differential equations to be investigated is already known. The second proposition then states that the problem of finding inequivalent cases of such symmetry extensions of the kernel algebra is reduced to the algebraic problem of the classification of subalgebras of the equivalence algebra.

\begin{remark}
Within the set $\mathcal S_i$ defined in Proposition~\ref{pro:OnPreliminaryGroupClassification2}, there is an equivalence relation generated by transformations from~$G^\sim$ whose push-forwards to vector fields preserve the subalgebra~$\mathfrak a_i$ of~$\mathfrak g^\sim$. Such transformations form the normalizer of the subgroup of~$G^\sim$ associated with~$\mathfrak a_i$. This equivalence relation can be used to choose simple forms of representatives of the set~$\mathcal S_i$.
\end{remark}

This now completes the picture of the methods available for general group classification problems. It should be clear that these methods apply to different classes of differential equations. This is why it is essential to investigate properties of the given class \emph{before} choosing a particular method of group classification. This is done in the present paper. It is shown in the subsequent sections that the class~\eqref{eq:GenDiffEqs} is not normalized. Moreover, a compatibility analysis of the determining equations of Lie symmetries of this class is also an overly complicated task. This is why it cannot be expected to solve the complete group classification problem for~\eqref{eq:GenDiffEqs} in a reasonable way. Still, the given class is adequate to be investigated using the method of complete preliminary group classification.

Recall that, as mentioned in the introduction, the class~\eqref{eq:GenDiffEqs} is contained in the wider class of equations of the general form $u_t=F(t,x,u,u_x)u_{xx}+G(t,x,u,u_x)$, which is normalized and for which the group classification problem was solved in~\cite{basa01Ay}. 

The following folklore assertion is true. 

\begin{proposition}\label{pro:OnKernelGroupAsANormalSubgroup}
The kernel (common part) $G^\cap=\bigcap_{\theta\in{\mathcal S}}G_\theta$ of  
the maximal point symmetry groups~$G_\theta$, $\theta\in\mathcal S$, of systems from the class~$\mathcal L|_{\mathcal S}$ 
is naturally embedded into the (usual) equivalence group~$G^{\sim}$ of this class via trivial (identical) prolongation of
the kernel transformations to the arbitrary elements.
The associated subgroup $\hat G^\cap$ of $G^{\sim}$ is normal.
\end{proposition}

\begin{proof}
Let $\mathcal T_0$ be an arbitrary element of $G^\cap$, i.e.\ $\mathcal T_0$ is a point symmetry transformation for any equation from the class~$\mathcal L|_{\mathcal S}$. Denote by~$\hat{\mathcal T}_0$ the trivial prolongation of~$\mathcal T_0$ to the arbitrary elements~$\theta$, $\hat{\mathcal T}_0\theta=\theta$. The transformation~$\hat{\mathcal T}_0$ obviously belongs to~$G^\sim$, since it maps any equation from~$\mathcal L|_{\mathcal S}$ to the same equation in the new variables and therefore saves the entire class~$\mathcal L|_{\mathcal S}$.

Taking an arbitrary transformation~$\mathcal T \in G^\sim$, consider the composition $\mathcal T^{-1}\hat{\mathcal T}_0\mathcal T$. In order to check that $\hat G^\cap$ is a normal subgroup of~$G^\sim$, we should prove that this composition belongs to~$\hat G^\cap$. We fix any~$\theta\in\mathcal S$ and denote $\mathcal T\theta$ by $\tilde\theta$. Then $\hat{\mathcal T}_0\mathcal T\theta = \tilde\theta$ and hence~$\mathcal T^{-1}\hat{\mathcal T}_0\mathcal T\theta = \theta$. This means that the projection ${\rm P}\mathcal T^{-1}\hat{\mathcal T}_0\mathcal T$ to the space of variables $(x,u)$ is a point symmetry transformation of $\mathcal L_\theta$ for any $\theta\in\mathcal S$. In other words, the transformation ${\rm P}\mathcal T^{-1}\hat{\mathcal T}_0\mathcal T$ is an element of~$G^\cap$. Therefore, $\mathcal T^{-1}\hat{\mathcal T}_0\mathcal T$, which is the trivial prolongation of ${\rm P}\mathcal T^{-1}\hat{\mathcal T}_0\mathcal T$ to the arbitrary elements, belongs to~$\hat G^\cap$.
\end{proof}

Properties of~$G^\cap$ described in Proposition~\ref{pro:OnKernelGroupAsANormalSubgroup} were first noted in different works by Ovsiannikov 
(see, e.g., \cite{ovsi75Ay} and \cite[Section~II.6.5]{ovsi82Ay}). 
Another formulation of this proposition is given in \cite[p.~52]{lisl92Ay}, Proposition~3.3.9.

\begin{corollary}\label{cor:OnKernelAlgebraAsIdeal1}
The trivial prolongation~$\hat{\mathfrak g}^\cap$ of the kernel algebra~$\mathfrak g^\cap$ to the arbitrary elements is an ideal in the equivalence algebra~$\mathfrak g^\sim$. 
\end{corollary}

By definition, any element of the algebra~$\hat{\mathfrak g}^\cap$ formally has the same form as the associated element from~$\mathfrak g^\cap$, but in fact is a vector field in the different space augmented with the arbitrary elements.

\begin{proof}
Consider arbitrary vector fields $Q_0\in\mathfrak g^\cap$ and $Q\in\mathfrak g^\sim$. Denote the trivial prolongation of~$Q_0$ to the arbitrary elements by~$\hat Q_0$, so $\hat Q_0\in\hat{\mathfrak g}^\cap$. It is necessary only to prove that $\hat Q_0\in\mathfrak g^\sim$ and $[Q,\hat Q_0]\in\hat{\mathfrak g}^\cap$. Let $\hat G_0=\{\hat{\mathcal T}_0(\ve)=\exp(\ve\hat Q_0)\}$ and $G=\{\mathcal T(\ve)=\exp(\ve Q)\}$ be local one-parameter transformation groups associated with $\hat Q_0$ and $Q$, respectively. As $\hat G_0$ is a subgroup of~$G^\sim$, the vector field $\hat Q_0$ belongs to~$\mathfrak g^\sim$.

For each sufficiently small $\ve$ define the composition $\tilde{\mathcal T}(\ve)=\hat{\mathcal T_0}(-\sqrt{\ve})\mathcal T(-\sqrt{\ve})\hat{\mathcal T_0}(\sqrt{\ve})\mathcal T(\sqrt{\ve})$
and consider the vector field
\[
    \tilde Q = \dd{}{\ve}\bigg|_{\ve=0+}\tilde{\mathcal T}(\ve),
\]
which coincides with $[Q,\hat Q_0]$, see e.g.~\cite[Theorem 1.33]{olve86Ay}. As both $\mathcal T(-\sqrt{\ve})\hat{\mathcal T_0}(\sqrt{\ve})\mathcal T(\sqrt{\ve})$ and $\hat{\mathcal T_0}(-\sqrt{\ve})$ belong to~$\hat G_0$ (cf. Proposition~\ref{pro:OnKernelGroupAsANormalSubgroup}), the transformation $\tilde{\mathcal T}(\ve)$ also is an element of~$\hat G_0$. Therefore, $\tilde Q\in \hat{\mathfrak g}^\cap$.
\end{proof}

As the kernel is included in the maximal Lie invariance algebra of any equation from the class, we should classify only subalgebras of the equivalence algebra that contain the ideal associated with the kernel.

\begin{example}
In general, the kernel~$\mathfrak g^\cap$ is not necessarily an ideal of the maximal Lie invariance algebra~$\mathfrak g_\theta$ for each $\theta\in\mathcal S$. 
Indeed, consider the class of $(1+1)$-dimensional nonlinear diffusion equations of the general form $u_t=(F(u)u_x)_x$, where $F\ne0$, cf.\ the introduction. 
The kernel of this class and the maximal Lie invariance algebra of the diffusion equation with $F=u^{-4/3}$ are 
$\mathfrak g^\cap=\langle\p_t,\,\p_x,\,2t\p_t+x\p_x\rangle$ and  
$\mathfrak g_1=\langle\p_t,\,\p_x,\,2t\p_t+x\p_x,\,4t\p_t+3u\p_u,\,x^2\p_x-3xu\p_u\rangle$
\cite{akha91Ay,ovsi59Ay,ovsi82Ay}, respectively. 
At the same time, the kernel~$\mathfrak g^\cap$ is not an ideal of~$\mathfrak g_1$, $[\mathfrak g^\cap,\mathfrak g_1]\not\subset\mathfrak g^\cap$, since 
\[
[\p_x,x^2\p_x-3xu\p_u]=2x\p_x-3xu\p_u\not\in\mathfrak g^\cap.
\]
Note that the class of diffusion equations is semi-normalized (see \cite{popo06Ay,popo10Ay} for the definition of semi-normalization) but not normalized in the usual sense.
\end{example}

\begin{corollary}
If the class $\mathcal L|_{\mathcal S}$ is normalized in the usual sense, 
the kernel algebra~$\mathfrak g^\cap$ is an ideal of the maximal Lie invariance algebra~$\mathfrak g_\theta$ for each $\theta\in\mathcal S$. 
\end{corollary}

\begin{proof}
We fix an arbitrary element $\theta^0\in\mathcal S$. 
Denote by~$\hat{\mathfrak g}_{\theta^0}$ the maximal subalgebra of~$\mathfrak g^\sim$ 
such that the algebraic equation $\theta=\theta^0(x,u_{(r)})$ is invariant with respect to~it. 
This subalgebra necessarily contains the trivial prolongation~$\hat{\mathfrak g}^\cap$ of the kernel algebra~$\mathfrak g^\cap$ to the arbitrary elements. 
Thus, we have that $\hat{\mathfrak g}^\cap\subset\hat{\mathfrak g}_{\theta^0}\subset\mathfrak g^\sim$ and, in view of Corollary~\ref{cor:OnKernelAlgebraAsIdeal1}, 
$\hat{\mathfrak g}^\cap$ is an ideal in~$\mathfrak g^\sim$. 
Therefore, $\hat{\mathfrak g}^\cap$ is an ideal in~$\hat{\mathfrak g}_{\theta^0}$. 
As the class $\mathcal L|_{\mathcal S}$ is normalized in the usual sense, the projection of~$\hat{\mathfrak g}_{\theta^0}$ to the space of the variables~$(x,u)$ 
coincides with the maximal Lie invariance algebra~$\mathfrak g_{\theta^0}$ of the equation~$\mathcal L_{\theta^0}$.
By the construction, the projection of~$\hat{\mathfrak g}^\cap$ to the space of the variables~$(x,u)$ coincides with~$\mathfrak g^\cap$. 
Hence $\mathfrak g^\cap$ is an ideal in~$\mathfrak g_{\theta^0}$. 
\end{proof}

Often the equivalence algebra can be represented as a semi-direct sum of the ideal associated with the kernel and a certain subalgebra. To obtain preliminary group classification in this case, we in fact need to classify only inequivalent subalgebras of the complement of the kernel ideal. Projections of these subalgebras to the space of equation variables will give all possible inequivalent extensions of the kernel.

\begin{example}
We present a class of differential equations for which the above representation is not possible.
This is the class of $(1+1)$-dimensional linear second order homogeneous evolution equations which has the general form
\begin{equation}\label{EqGenLPE}
u_t=A(t,x)u_{xx}+B(t,x)u_x+C(t,x)u,
\end{equation}
where $A=A(t,x)$, $B=B(t,x)$ and $C=C(t,x)$ are arbitrary smooth functions, $A\ne0$.
The kernel Lie algebra of class~\eqref{EqGenLPE} is $\mathfrak g^\cap=\langle u\p_u\rangle$.
Its equivalence algebra~$\mathfrak g^\sim$ is spanned by operators of the form
{\samepage 
\begin{gather*}
\tau\p_t+\xi\p_x+\eta^1u\p_u+{} \\ (2\xi_x-\tau_t)\p_A+((\xi_x-\tau_t)B-2\eta^1_xA-\xi_t)\p_B+(\eta^1_t-A\eta^1_{xx}-B\eta^1_x-C\xi_t)\p_C,
\end{gather*}}
where $\tau=\tau(t)$, $\xi=\xi(t,x)$ and $\eta^1=\eta^1(t,x)$ are arbitrary smooth functions of their arguments.
The kernel~$\mathfrak g^\cap$ can be identified with the ideal~$\hat{\mathfrak g}^\cap$ of~$\mathfrak g^\sim$, generated by the vector field $u\p_u$, which is assumed now to act in the space of variables and arbitrary elements. Moreover, this vector field commutes with all elements of~$\mathfrak g^\sim$. At the same time we have $[\p_t,tu\p_u] = u\p_u$. Therefore the algebra~$\mathfrak g^\sim$ cannot be represented as a semi-direct sum of~$\hat{\mathfrak g}^\cap$ and a subalgebra.
\end{example}

\section{Determining equations of Lie symmetries}\label{sec:DeterminingEquationsOfLieSymmetries}

The method of computing Lie symmetries is classical and can be found in all textbooks on this subject, see, e.g.~\cite{blum89Ay,olve86Ay,ovsi82Ay}. Owing to its algorithmic nature, it was implemented in a number of symbolic computation programs~\cite{butc03Ay,carm00Ay,head93Ay,roch10Ay}. 
For an equation $\Delta=0$ from the class~\eqref{eq:GenDiffEqs}, the condition of infinitesimal invariance with respect to a vector field
\[
 Q = \tau(t,x,u)\p_t + \xi(t,x,u)\p_x + \eta(t,x,u)\p_u
\]
has the form $Q^{(2)}\Delta|_{\Delta=0}^{} = 0$, i.e., 
\begin{equation}\label{eq:InfinitesimalInvarianceCondition}
 Q^{(2)}\Delta = \eta^t - \xi f_x u_x^2 -\eta f_u u_x^2 - 2fu_x\eta^x - \xi g_xu_{xx} - \eta g_u u_{xx} - g\eta^{xx} = 0
\end{equation}
wherever $\Delta=0$. Here $Q^{(2)}$ is the second prolongation of the vector field~$Q$, 
\begin{equation}\label{eq:2ndProlongationOfQ}
Q^{(2)} = Q + \eta^t\p_{u_t} + \eta^x\p_{u_x} + \eta^{tt}\p_{u_{tt}} + \eta^{tx}\p_{u_{tx}} + \eta^{xx}\p_{u_{xx}},
\end{equation}
where the coefficients can be determined by using the general prolongation formula. In~\eqref{eq:InfinitesimalInvarianceCondition} we only need the coefficients~$\eta^t$, $\eta^x$ and $\eta^{xx}$. They read~\cite{olve86Ay,ovsi82Ay}
\begin{align}\label{eq:CoefficientsOfProlongation}
\begin{split}
 &\eta^t = \DD_t(\eta - \tau u_t - \xi u_x) + \tau u_{tt} + \xi u_{tx}, \\
 &\eta^x = \DD_x(\eta - \tau u_t - \xi u_x) + \tau u_{tx} + \xi u_{xx}, \\
 &\eta^{xx} = \DD^2_x(\eta - \tau u_t - \xi u_x) + \tau u_{txx} + \xi u_{xxx},
\end{split}
\end{align}
where $\DD_t$ and $\DD_x$ denote the operators of total differentiation with respect to $t$ and $x$, respectively, 
\[
 \DD_t = \p_t + u_t\p_u + u_{tt}\p_{u_t}+ u_{tx}\p_{u_x} + \cdots, \qquad \DD_x = \p_x + u_x\p_u + u_{tx}\p_{u_t}+ u_{xx}\p_{u_x} + \cdots.
\]

Upon plugging the coefficients~\eqref{eq:CoefficientsOfProlongation} into the infinitesimal invariance condition~\eqref{eq:InfinitesimalInvarianceCondition}, we obtain the following equation
\begin{align}\label{eq:DeterminingEquationUnsplit}
\begin{split}
 &\DD_t\eta - u_t\DD_t\tau - u_x\DD_t\xi-\xi f_xu_x^2 - \eta f_u u_x^2 - 2fu_x(\DD_x\eta -u_t\DD_x\tau - u_x\DD_x\xi)- \xi g_xu_{xx} -  {} \\
 &{} \eta g_u u_{xx} - g(\DD_x^2\eta - u_t \DD_x^2\tau - u_x \DD_x^2\xi - 2u_{tx}\DD_x\tau - 2u_{xx}\DD_x\xi) = 0.
\end{split}
\end{align}
In order to constrain this equation on the manifold of equations~\eqref{eq:GenDiffEqs}, we set $u_t = fu_x^2+gu_{xx}$. Then, splitting~\eqref{eq:DeterminingEquationUnsplit} with respect to the various derivatives of $u$ we obtain the following overdetermined system of determining equations of Lie symmetries:
\begin{align}\label{eq:DeterminingEquationSplit}
\begin{split}
 &u_xu_{xt}\colon \qquad \tau_u = 0, \\
 &u_{xt}\colon \hspace{1.18cm} \tau_x = 0, \\
 &u_xu_{xx}\colon \hspace{0.72cm} \xi_u = 0, \\
 &u_x^2\colon \hspace{1.28cm} f(\tau_t+\eta_u - 2\xi_x) + g\eta_{uu} + \xi f_x + \eta f_u = 0, \\
 &u_{xx}\colon \hspace{1.12cm} g(\tau_t - 2\xi_x) + \xi g_x + \eta g_u = 0, \\
 &u_x\colon \hspace{1.28cm} \xi_t + 2f\eta_x + g(2\eta_{xu}-\xi_{xx}) = 0, \\
 &1\colon \hspace{1.5cm} \eta_t - g\eta_{xx} = 0.
\end{split}
\end{align}
As usual for classes of differential equations, the determining equations split into a part not involving the arbitrary elements and a part explicitly involving them (the \textit{classifying part}). In the present case, the first three equations do not involve $f$ and $g$ and can therefore be integrated immediately. They give $\tau = \tau(t)$ and $\xi = \xi(t,x)$, i.e.\ the symmetry transformations are projectable and transformations of $t$ only depend on $t$.

The remaining four equations form the system of classifying equations. In the case of arbitrariness of the functions $f$ and $g$, we can further split system~\eqref{eq:DeterminingEquationSplit} with respect to derivatives of $f$ and $g$. This yields the kernel of maximal Lie invariance algebras, which gives rise to those symmetry transformations that are admitted for all elements of the class of equations~\eqref{eq:GenDiffEqs}. Splitting yields
\[
  \xi = \eta = 0, \quad \tau_t = 0,
\]
i.e.\ the \emph{kernel algebra}~$\mathfrak g^\cap$ is generated solely by the operator $\p_t$, $\mathfrak g^\cap=\langle\p_t\rangle$. That is, for arbitrary values of $f$ and $g$, the only symmetry admitted by equations of the class~\eqref{eq:GenDiffEqs} is the time translation symmetry $(t,x,u)\mapsto (t+\ve,x,u)$, $\ve\in\mathbb{R}$.

\section{The equivalence algebra}\label{sec:EquivalenceAlgebra}

In order to investigate inequivalent cases of symmetry extensions of the kernel algebra~$\mathfrak g^\cap$, the equivalence algebra (group) must be computed. Unfortunately, the equivalence algebra presented in~\cite{nadj10Ay} is not correct. It can easily be checked that their operator $b(x)\p_u$ does not generate an equivalence transformation for general values of $b$. Similarly, also their operator $a(x)\p_x + 2fa(x)\p_f+ga'(x)\p_g$ cannot generate equivalence transformations for arbitrary values of $a$. The problem indeed is that their infinitesimal invariance condition for equivalence transformations is incorrect. This is why it is necessary to re-derive the equivalence algebra for the class of equations~\eqref{eq:GenDiffEqs} here.

\begin{theorem}\label{thm:EquivalenceAlgebra}
The equivalence algebra~$\mathfrak g^\sim$ of the class of equations~\eqref{eq:GenDiffEqs} is generated by the following operators:
\begin{gather}\label{eq:EquivalenceAlgebra}
\begin{split}
 &\p_t,\quad \p_x,\quad \DDD^t= t\p_t - f\p_f-g\p_g, \quad \DDD^x=x\p_x+2f\p_f+2g\p_g,\\
 &\GG(h)=h\p_u - (h_uf+h_{uu}g)\p_f,
\end{split}
\end{gather}
where $h=h(u)$ is an arbitrary smooth function of $u$.
\end{theorem}

\begin{proof}
The proof is done using infinitesimal methods. We seek for operators of the form
\[
  Y = \tau\p_t + \xi\p_x + \eta\p_u + \varphi\p_f + \theta\p_g
\]
that generate continuous equivalence transformations, where $\tau$, $\xi$ and $\eta$ are functions of the variables $t$, $x$ and $u$, whereas $\varphi$ and $\theta$ are regarded as functions of $t$, $x$, $u$, $f$ and $g$. That is, we aim to determine the usual equivalence algebra rather than some generalized equivalence algebra~\cite{mele96Ay,popo10Ay}. The class of equations~\eqref{eq:GenDiffEqs} must be augmented with the auxiliary system
\begin{equation}\label{eq:AuxiliarySystem}
 S_1:= f_t = 0, \quad S_2:= g_t = 0.
\end{equation}
The complete auxiliary system should also include the conditions that the arbitrary elements~$f$ and~$g$ do not depend on nonzero order derivatives of~$u$. 
However, these conditions already are implicitly taken into account by the supposition that the coefficients of~$Y$ does not involve these derivatives. 

The joint invariance condition then reads
\begin{equation}\label{eq:InfinitesimalInvarianceConditionEquivalenceAlgebra}
 \tilde Y\Delta\big|_{\mathcal M} = 0, \quad \tilde YS_1\big|_{\mathcal M} = 0, \quad \tilde YS_2\big|_{\mathcal M} = 0,
\end{equation}
where $\mathcal M$ denotes the joint system of the equations $\Delta=0$, $S_1=0$ and $S_2=0$, 
\[
 \tilde Y = Q^{(2)} + \varphi\p_f + \theta\p_g + \varphi^t\p_{f_t} + \theta^t\p_{g_t}, 
\]
and $Q^{(2)}$ is defined by~\eqref{eq:2ndProlongationOfQ}.
The coefficients $\varphi^t$ and $\theta^t$ can be obtained by the first prolongation considering 
$(t,x,u)$ and $(f,g)$ as independent and dependent variables, respectively,  
\begin{align*}
 &\varphi^t = \tilde\DD_t(\varphi - \tau f_t - \xi f_x - \eta f_u) + \tau f_{tt} + \xi f_{tx} + \eta f_{tu}, \\
 &\theta^t = \tilde\DD_t(\varphi - \tau g_t - \xi g_x - \eta g_u) + \tau g_{tt} + \xi g_{tx} + \eta g_{tu},
\end{align*}
where $\tilde\DD_t = \p_t + f_t\p_f + g_t\p_g + \cdots$ is the corresponding operator of total differentiation with respect to~$t$. In view of the auxiliary system~\eqref{eq:AuxiliarySystem}, the total derivative operator reduces to the partial derivative, i.e.\ $\tilde\DD_t = \p_t$.

The second and the third conditions from~\eqref{eq:InfinitesimalInvarianceConditionEquivalenceAlgebra} then imply that
\[
 \varphi_t - \xi_tf_x - \eta_tf_u = 0, \quad \theta_t - \xi_tg_x - \eta_tg_u = 0.
\]
Since these equations should be satisfied for all values of the arbitrary elements $f$ and $g$, 
we can split with respect to the derivatives $f_x$, $f_u$, $g_x$ and $g_u$ to obtain that
\[
 \varphi_t = \theta_t = \xi_t = \eta_t = 0.
\]
It remains to investigate the first condition in~\eqref{eq:InfinitesimalInvarianceConditionEquivalenceAlgebra}. In detail, it reads
\[
 \eta^t - 2fu_x\eta^x - \varphi u_x^2 - g\eta^{xx} - \theta u_{xx} = 0,
\]
or, after expanding,
\begin{align*}
  &\DD_t\eta - u_t\DD_t\tau - u_x\DD_t\xi - 2fu_x(\DD_x\eta - u_t\DD_x\tau-u_x\DD_x\xi) - \varphi u_x^2 - {} \\
  & {} g(\DD_x^2\eta - u_t\DD_x^2\tau - u_x\DD_x^2\xi - 2u_{tx}\DD_x\tau - 2u_{xx}\DD_x\xi) - \theta u_{xx} = 0.
\end{align*}
We now split this equation with respect to the derivatives of $u$ similar as done in the course of deriving the determining equations of Lie symmetries. The splitting with respect to $u_{tx}$ implies that $\tau=\tau(t)$. Splitting with respect to $u_xu_{xx}$, we derive that $\xi=\xi(x)$. These conditions already simplifies the above invariance condition substantially. Collecting coefficients of the remaining monomials of derivatives leads to
\begin{align*}
 &u_{xx}\colon \qquad \theta= (2\xi_x-\tau_t)g, \\
 &u_x^2\colon \hspace{0.94cm}  \varphi  = (2\xi_x - \tau_t-\eta_u)f - \eta_{uu} g, \\
 &u_x\colon \hspace{0.94cm} 2\eta_xf + 2\eta_{xu}g - \xi_{xx}g = 0, \\
 &1\colon \hspace{1.16cm} g\eta_{xx} = 0.
\end{align*}
In view of $\varphi_t = \theta_t = \xi_t = \eta_t = 0$, the general solution of this system is
\begin{align*}
 &\tau = c_1t+c_2, \quad \xi = c_3x + c_4,\quad \eta = h(u),\\
 &\varphi = (2c_3-c_1-h_u)f-h_{uu} g,\quad \theta = (2c_3-c_1)g,
\end{align*}
where $c_1$, \dots, $c_4$ are arbitrary constants and $h$ is an arbitrary smooth function of~$u$.

This completes the proof of the theorem.
\end{proof}

The equivalence algebra $\mathfrak g^\sim$ can be represented in several ways, which are important for different purposes. The representation crucial for the present case is that $\mathfrak g^\sim=\langle\p_t\rangle\lsemioplus\langle\DDD^x,\DDD^t,\p_x,\GG(h)\rangle$, see Remark~\ref{rem:StructureOfEquivalenceAlgebra} for further details. Another natural representation is $\mathfrak g^\sim = \langle\p_t\lsemioplus\DDD^t\rangle\oplus\langle\p_x\lsemioplus\DDD^x\rangle\oplus\langle\GG(h)\rangle$. This  representation implies that $\mathfrak g^\sim$ is the direct sum of a finite-dimensional and an infinite-dimensional parts. This representation is helpful for the determination of the adjoint actions, see Section~\ref{sec:PreliminaryGroupClassification}.

\section{The equivalence group}\label{sec:EquivalenceGroup}

In the previous section we have determined the equivalence algebra of the class~\eqref{eq:GenDiffEqs} using infinitesimal techniques. In order to obtain the complete point equivalence group (including also discrete transformations), the direct method should be applied. For the sake of completeness and illustration we present the corresponding computations here.

\begin{theorem}\label{thm:EquivGroupOfGenDiffEqs}
The equivalence group~$G^\sim$ of the class of equations~\eqref{eq:GenDiffEqs} is formed by the transformations
\begin{align*}
 &\tilde t = A_1t+A_0, \quad \tilde x = B_1x+B_0, \quad \tilde u = U(u), \\
 &\tilde f = \frac{B_1^2}{A_1U_u}\left(f-\frac{U_{uu}}{U_u}g\right), \quad \tilde g = \frac{B_1^2}{A_1}g,
\end{align*}
where $A_0, A_1, B_0, B_1\in\mathbb{R}$, $U$ is an arbitrary smooth function of $u$ and $A_1B_1U_u\ne0$.
\end{theorem}

\begin{proof}
We begin with a preliminary description of admissible transformations of the class~\eqref{eq:GenDiffEqs}.
In other words, we derive determining equations for point transformations that map a fixed equation from the class~\eqref{eq:GenDiffEqs} to an equation from the same class. As~\eqref{eq:GenDiffEqs} defines a subclass of $(1+1)$-dimensional evolution equations, we at once know that the transformation component of~$t$ depends only on~$t$, see e.g.~\cite{king98Ay,maga93Ay}. Moreover, each equation from the class~\eqref{eq:GenDiffEqs} belongs to 
the class of second-order quasi-linear evolution equations having the form $u_t=F(t,x,u)u_{xx}+G(t,x,u,u_x)$. 
Hence in view of Lemma~1 of~\cite{ivan10Ay} the transformation component of~$x$ depends only on~$t$ and~$x$.
That is, the transformations of variables will be of the form $\tilde t = T(t)$, $\tilde x=X(t,x)$, $\tilde u = U(t,x,u)$, where $T_tX_xU_u\ne0$. The transformed derivatives then read
\[
 \tilde u_{\tilde t} = \frac{1}{T_t}\left(\DD_tU - \frac{X_t}{X_x}\DD_xU\right), \quad \tilde u_{\tilde x} = \frac1{X_x}\DD_xU, \quad \tilde u_{\tilde x\tilde x} = \left(\frac1{X_x}\DD_x\right)^2U.
\]
Substituting these derivatives into the transformed form of~\eqref{eq:GenDiffEqs} and taking into account the initial form~\eqref{eq:GenDiffEqs}, we obtain 
\begin{equation}\label{eq:DeterminingEquationsFiniteEquivalenceGroup}
\frac{1}{T_t}\left(U_t - \frac{X_t}{X_x}\DD_xU\right) + \frac{U_u}{T_t}(fu_x^2+gu_{xx})= 
\tilde f\left(\frac{\DD_xU}{X_x}\right)^2 + \tilde g \left(\frac1{X_x}\DD_x\right)^2U,
\end{equation}
where $\tilde f=\tilde f(X,U)$ and $\tilde g=\tilde g(X,U)$.
Splitting equation~\eqref{eq:DeterminingEquationsFiniteEquivalenceGroup} with respect to $u_{xx}$ and $u_x$ yields
\begin{subequations}\label{eq:DeterminingEquationsAdmissibleTransformations}
\begin{gather}
  \label{eq:DeterminingEquationsAdmissibleTransformationsA}
 \hbox to 13mm {$u_{xx}\colon$\hfil} \tilde g = \frac{X_x{}^2}{T_t}g,
\\\label{eq:DeterminingEquationsAdmissibleTransformationsB}
 \hbox to 13mm {$u_x^2\colon$\hfil}  f\frac{U_u}{T_t} = \tilde f\frac{U_u^2}{X_x^2}+\tilde g\frac{U_{uu}}{X_x^2}, 
\\\label{eq:DeterminingEquationsAdmissibleTransformationsC}
 \hbox to 12mm {$u_x\colon$\hfil}    -\frac{X_tU_u}{T_tX_x}=\frac{\tilde g}{X_x^2}\left(2U_{ux}-\frac{X_{xx}}{X_x}U_u\right)+2\tilde f\frac{U_xU_u}{X_x^2}, 
\\\label{eq:DeterminingEquationsAdmissibleTransformationsD}
 \hbox to 13mm {$1\colon$\hfil}      \frac1{T_t}\left(U_t-\frac{X_t}{X_x}U_x\right)=\tilde f\left(\frac{U_x}{X_x}\right)^2+\frac{\tilde g}{X_x^2}\left( U_{xx}-\frac{X_{xx}}{X_x}U_x\right).
\end{gather}
\end{subequations}
Equation~\eqref{eq:DeterminingEquationsAdmissibleTransformationsA} defines the transformation rule for the arbitrary element~$g$.
Substituting equation~\eqref{eq:DeterminingEquationsAdmissibleTransformationsA} into~\eqref{eq:DeterminingEquationsAdmissibleTransformationsB} leads to the transformation rule for the arbitrary element~$f$,
\[
 \tilde f = \frac{X_x^2}{T_tU_u}f - \frac{X_x^2U_{uu}}{T_tU_u^2}g.
\]
In general, system~\eqref{eq:DeterminingEquationsAdmissibleTransformations} forms the determining equations for admissible transformations of the class~\eqref{eq:GenDiffEqs}. Unfortunately, this system is too difficult to be integrated since there are a lot of different cases of its solution depending on specific values of the arbitrary elements. However, this system allows to easily determine the equivalence group. For this aim, we can split 
equations~\eqref{eq:DeterminingEquationsAdmissibleTransformationsC} and~\eqref{eq:DeterminingEquationsAdmissibleTransformationsD}
with respect to~$\tilde f$ and~$\tilde g$. This gives at once $X_t=X_{xx}=U_x=U_t=0$ since $U_u\ne0$. Furthermore, differentiating the first equation of system~\eqref{eq:DeterminingEquationsAdmissibleTransformations} with respect to $t$ leads to the final restriction $T_{tt}=0$. Solving these determining equations for equivalence transformations completes the proof of the theorem.
\end{proof}

\begin{corollary}
A complete set of discrete equivalence transformations in the group~$G^\sim$, 
which are independent up to their composition and composition with continuous transformations 
are exhausted by the three transformations of alternating signs
\[\arraycolsep=0ex
\begin{array}{ll}
I_t\colon\qquad&(t,x,u,f,g)\mapsto(-t,x,u,-f,-g),\\
I_x\colon\qquad&(t,x,u,f,g)\mapsto(t,-x,u,f,g),\\
I_u\colon\qquad&(t,x,u,f,g)\mapsto(t,x,-u,-f,g).
\end{array}
\]
\end{corollary}

The equation~\eqref{eq:GenDiffEqs} with the specific value $\theta_0=(f,g)=(-4/3u^{-7/3},u^{-4/3})$ admits the Lie symmetry operator $x^2\p_x-3xu\p_u$. The transformations from the corresponding one-parameter transformation group belong to~$\mathrm T(\theta_0,\theta_0)$. As the associated admissible transformations are not induced by elements of the equivalence group~$G^\sim$ of the class~\eqref{eq:GenDiffEqs}, this class is not normalized. 
Similar assertions are true for the potential Burgers equations ($f=g=1$), linear equations from the class~\eqref{eq:GenDiffEqs} ($f=0$, $g_u=0$), etc.
As system~\eqref{eq:DeterminingEquationsAdmissibleTransformations} is too complicated and the equivalence group~$G^\sim$ is quite narrow in comparison with the class~\eqref{eq:GenDiffEqs} (the transformations from~$G^\sim$ are parameterized by four constants and only a singe function of one argument and, at the same time, the tuple of arbitrary elements consists of two functions of two arguments), this justifies why preliminary group classification is well suited for the class of equations~\eqref{eq:GenDiffEqs}.

\begin{remark}\label{rem:OnParticularSiplificationOfGenDiffEqs}
It is not possible to simplify the general equation from the class~\eqref{eq:GenDiffEqs} by equivalence transformations.
The interesting particular case of simplification by equivalence transformations is given by equations of the form~\eqref{eq:GenDiffEqs} with $f$ proportional to~$g$.
If $f=cg$, where $c$ is a nonzero constant, then the corresponding equation of the form~\eqref{eq:GenDiffEqs} is mapped by the transformation 
\begin{equation}\label{eq:TransForGenDiffEqsWithFEqualG}
\tilde t =t,\ \tilde x=x,\ \tilde u=e^{cu}
\end{equation}
to the equation of the same form with $\tilde f=0$ and $\tilde g=g(\tilde x,c^{-1}\ln\tilde u)$.
\end{remark}

\section{Classification of subalgebras}\label{sec:ClassificationOfSubalgebras}

In order to carry out preliminary group classification, it is necessary to derive an optimal list of inequivalent subalgebras. In the existing literature on the subject, usually only subalgebras of a certain finite-dimensional subalgebra of the equivalence algebra are classified up to inner automorphisms of this subalgebra. This restriction is, however, not necessary in the present case, although this is done in~\cite{nadj10Ay}. Furthermore, it should be noted that in~\cite{nadj10Ay} an algebra was chosen for preliminary group classification, which is not related to the corresponding equivalence algebra that was derived.

To classify subalgebras of a Lie algebra of vector fields, it is necessary to know the adjoint action of the corresponding transformation (pseudo)\-group on this algebra. There exist two different methods for the computation of the adjoint action. The first method employs information on the structure of the Lie algebra and is more suitable in the finite-dimensional case although it also works for certain infinite-dimensional algebras~\cite{bihl09Ay,bihl10Dy,fush94Ay,popo10Cy}.
The adjoint action of a one-parameter Lie group generated by an element~$\vv$ of the Lie algebra on this algebra can be determined either from the Lie series
\[
 \mathbf{w}(\ve) = \Ad(e^{\ve\vv})\mathbf{w}_0 := \sum_{n=0}^{\infty}\frac{\ve^n}{n!}\{\vv^n,\ww_0 \},
\]
where $\{\vv^0,\ww_0\} := \ww_0, \{\vv^n,\ww_0\} := (-1)^n[\vv,\{\vv^{n-1},\ww_0\}]$, or, by solving the Cauchy problem
\[
 \dd{\mathbf{w}}{\ve} = [\mathbf{w},\vv], \quad \mathbf{w}(0) = \mathbf{w}_0,
\]
see~\cite{olve86Ay} for more details.

Following the second method, it is necessary only to calculate the actions of push-forwards of the transformations from the (pseudo)group on generating vector fields of the algebra. While the first method involves only the abstract structure of Lie algebras and therefore at once gives results for the whole class of isomorphic algebras, the second method relies on the specific realization of the Lie algebra by vector fields. At the same time, the second method works more properly in the infinite-dimensional case.

We now derive the optimal lists of one- and two-dimensional subalgebras for the entire equivalence algebra~$\mathfrak g^\sim$. 

The nonzero commutation relations of generating elements~\eqref{eq:EquivalenceAlgebra} of~$\mathfrak g^\sim$ are
\[
 [\p_x,\DDD^x]=\p_x,\quad [\p_t,\DDD^t] = \p_t,\quad [\GG(h^1),\GG(h^2)] = \GG(h^1h^2_u-h^2h^1_u).
\]
The nonidentical adjoint actions related to generating elements of~$\mathfrak g^\sim$ and computed using the first method are 
\[\arraycolsep=0ex
\begin{array}{lll}
\Ad(e^{\ve\p_x})\DDD^x = \DDD^x - \ve\p_x, \qquad& \Ad(e^{\ve\DDD^x})\p_x = e^{\ve}\p_x,\qquad&\Ad({e^{\ve\GG(h^1)}})\GG(h^2) = \GG(\tilde h^2), \\
\Ad(e^{\ve\p_t})\DDD^t = \DDD^t - \ve\p_t, \qquad& \Ad(e^{\ve\DDD^t})\p_t = e^{\ve}\p_t,\\
 
\end{array}
\]
where $\tilde h^2(u,\ve)=h^2(H^1(u,-\ve))/H^1_u(u,-\ve)$ and 
$\{\tilde u = H^1(u,\ve)\}$ is the one-parameter transformation group generated by the projection of the operator $\GG(h^1)$ to the space of the variable~$u$, i.e.~$H^1_\ve = h^1(H^1)$ and $H^1(u,0)=u$. 
Although the four adjoint actions related to the finite-dimensional part of~$\mathfrak g^\sim$ are suitable to be applied to the classification, 
there arises an inconvenience with the adjoint action $\Ad({e^{\ve\GG(h^1)}})$ owing to problems with proving the existence of the required function $h^1$.  

This is why, in what follows we use the adjoint action of the entire equivalence group~$G^\sim$ on the equivalence algebra~$\mathfrak g^\sim$, calculated by the second method. 
Any transformation $\mathcal T$ from~$G^\sim$ can be represented, for convenience,  as a composition 
\[
\mathcal T=\mathscr T^t(A_0)\mathscr T^x(B_0)\mathscr D^t(A_1)\mathscr D^x(B_1)\mathscr G(U),
\] 
cf.~Theorem~\ref{thm:EquivGroupOfGenDiffEqs}, where
\begin{equation}\label{eq:EquivGroupGeneratingTrans}
\arraycolsep=0ex
\begin{array}{llllll}
\mathscr T^t(A_0)\colon\quad& \tilde t=t+A_0,\quad& \tilde x=x,    \quad& \tilde u=u,   \quad& \tilde g=g,        \quad& \tilde f=f,                       \\
\mathscr T^x(B_0)\colon\quad& \tilde t=t,    \quad& \tilde x=x+B_0,\quad& \tilde u=u,   \quad& \tilde g=g,        \quad& \tilde f=f,                       \\
\mathscr D^t(A_1)\colon\quad& \tilde t=A_1t, \quad& \tilde x=x,    \quad& \tilde u=u,   \quad& \tilde g=A_1^{-1}g,\quad& \tilde f=A_1^{-1}f,               \\
\mathscr D^x(B_1)\colon\quad& \tilde t=t,    \quad& \tilde x=B_1x, \quad& \tilde u=u,   \quad& \tilde g=B_1^2g,   \quad& \tilde f=B_1^2f,                  \\
\mathscr G(U)    \colon\quad& \tilde t=t,    \quad& \tilde x=x,    \quad& \tilde u=U(u),\quad& \tilde g=g,        \quad& \tilde f=f/U_u-gU_{uu}/U_u^2
\end{array}
\end{equation}
are translations with respect to~$t$ and~$x$, scalings with respect to~$t$ and~$x$ and an arbitrary transformation of~$u$, respectively, and $A_1B_1U_u\ne0$.
Transformations of each of the above kinds form a subgroup of~$G^\sim$. 
The last three subgroups contain the discrete transformations~$I_t$, $I_x$ and~$I_u$, respectively. 
Namely, $I_t=\mathscr D^t(-1)$, $I_x=\mathscr D^x(-1)$ and $I_u=\mathscr G(-u)$. 
As a result, additionally to avoiding the above problems with the existence of required values of functional parameters, 
in this way we at once include discrete equivalence transformations in the classification procedure. 

The nonidentical actions of push-forwards of transformations~\eqref{eq:EquivGroupGeneratingTrans} on generating elements of $\mathfrak g^\sim$ 
are exhausted by the followings:
\[\arraycolsep=0ex
\begin{array}{lll}
\mathscr T^x_*(B_0)\DDD^x = \DDD^x-B_0\p_x, \qquad& \mathscr D^x_*(B_1)\p_x = B_1\p_x,  \qquad& \mathscr G_*(U)\GG(h) = \GG(h(\tilde U)/\tilde U_u),\\
\mathscr T^t_*(A_0)\DDD^t = \DDD^t-A_0\p_t, \qquad& \mathscr D^t_*(A_1)\p_t = A_1\p_t,
\end{array}
\]
where the function $\tilde U=\tilde U(u)$ is the inverse of~$U$.

\begin{remark}\label{rem:StructureOfEquivalenceAlgebra}
 The kernel algebra generated by $\p_t$ is an ideal in the equivalence algebra~$\mathfrak g^\sim$, which has the structure $\mathfrak g^\sim=\langle\p_t\rangle\lsemioplus \langle\DDD^x,\DDD^t,\p_x,\GG(h)\rangle$. Hence the classification of subalgebras of~$\mathfrak g^\sim$ can be reduced to the classification of subalgebras of the algebra $\mathfrak g^\sim_{\rm ess}=\langle\DDD^x,\DDD^t,\p_x,\GG(h)\rangle$, which is the ``essential'' part of~$\mathfrak g^\sim$. This will yield the possible Lie invariance algebra extensions of the kernel algebra obtainable by preliminary group classification. Moreover, the push-forwards of translations and scalings with respect to~$t$ should not be applied under the classification of subalgebras.
\end{remark}

\begin{theorem}\label{thm:InequivalentOneDimensionalSubalgebras}
 An optimal list of one-dimensional subalgebras of the algebra $\mathfrak g^\sim_{\rm ess}$ is exhausted by the algebras
\begin{equation}\label{eq:OptimalListOneDimension}
 \langle\DDD^x + a\DDD^t-\GG(\delta)\rangle, \quad \langle\DDD^t+\tilde\delta\p_x-\GG(\delta)\rangle, \quad \langle\p_x-\GG(\delta)\rangle,\quad \langle\GG(1)\rangle,
\end{equation}
where $a\in\mathbb R$ and $\delta,\tilde\delta\in\{0,1\}$.
\end{theorem}

\begin{proof}
We use the approach for the classification of subalgebras that is outlined in~\cite{olve86Ay}. 
We~start with the most general form of an element of the algebra~$\mathfrak g^\sim_{\rm ess}$,
\[
 \vv^1= a^1_1\DDD^x+a^1_2\DDD^t+a^1_3\p_x + \GG(h^1),
\]
where the constants $a^1_1$, $a^1_2$, $a^1_3$ and the function $h^1=h^1(u)$ are arbitrary but fixed, 
and simplify it as much as possible by means of push-forwards of transformations from the equivalence group~$G^\sim$. 
In the case $h^1\ne0$ the function-parameter~$h^1$ can be set to~$-1$ by usage of $\mathscr G_*(U)$ 
with the inverse~$U$ to a solution $\tilde U=\tilde U(u)$ of the equation $\tilde U_u=-h^1(\tilde U)$.
In other words, up to $G^\sim$-equivalence we can always assume that $-h^1=\delta\in\{0,1\}$.

If $a^1_1\ne0$, we scale~$\vv^1$ to set $a^1_1=1$ and use the push-forward of a $\mathscr T^x(B_0)$ to set $a^1_3=0$. 
The notation $a=a^1_2$ leads to the first subalgebra in the list~\eqref{eq:OptimalListOneDimension}. 

If $a^1_1=0$ and $a^1_2\ne0$, we set $a^1_2=1$ by scaling~$\vv^1$ and use $\mathscr D^x_*(B_1)$ with certain~$B_1$ 
to set $a^1_3=-\tilde\delta$, where $\tilde\delta\in\{0,1\}$. 
This gives the second listed subalgebra. 

In the remaining case $a^1_1=a^1_2=0$ we obtain the two last subalgebras from the list~\eqref{eq:OptimalListOneDimension} 
under the assumptions $a^1_3\ne0$ and $a^1_3=0$, respectively, since  
the nonvanishing value of~$a^1_3$ is set to be equal to~1 by scaling~$\vv^1$ and  
the condition $a^1_3=0$ necessarily implies that $h^1\ne0$ and hence, up to $G^\sim$-equivalence, $h^1=1$.
\end{proof}

\begin{theorem}\label{thm:InequivalentTwoDimensionalSubalgebras}
 An optimal list of two-dimensional subalgebras of the algebra $\mathfrak g^\sim_{\rm ess}$ reads
\begin{align}\label{eq:OptimalListTwoDimension}
\begin{split}
&\langle\DDD^x-\GG(\hat\delta), \DDD^t-\GG(\delta)\rangle,\quad 
 \langle\DDD^x+a\DDD^t+\GG(u), \p_x-\GG(1)\rangle,\quad 
 \langle\DDD^x+a\DDD^t-\GG(\delta), \p_x\rangle \\
&\langle\DDD^t-\GG(\delta),\p_x-\GG(\tilde\delta)\rangle,\quad 
 \langle\DDD^x+a\DDD^t+b\GG(u),\GG(1)\rangle,\quad 
 \langle\DDD^t-\delta\p_x+b\GG(u),\GG(1)\rangle,\\ 
&\langle\p_x-\delta\GG(u),\GG(1)\rangle,\quad
 \langle\GG(1),\GG(u)\rangle,
\end{split}
\end{align}
where $a$, $b$, $\delta$, $\tilde\delta$ and $\hat\delta$ are constants,
and we can assume that 
$\delta,\tilde\delta\in\{0,1\}$, $\hat\delta\in\mathbb R$ if $\delta=1$ and $\hat\delta\in\{0,1\}$ if $\delta=0$.
\end{theorem}

\begin{proof}
The proof of the above theorem is similar to those in the one-dimensional case, see a detailed explanation and other examples in~\cite[Chapter~7]{bihl10Ay}. We start with two linearly independent copies of the most general element of~$\mathfrak g^\sim$,
\begin{align*}
 &\vv^1=a^1_1\DDD^x+a^1_2\DDD^t+a^1_3\p_x+\GG(h^1), \\
 &\vv^2=a^2_1\DDD^x+a^2_2\DDD^t+a^2_3\p_x+\GG(h^2),
\end{align*}
and simplify them as much as possible by means of adjoint actions and nondegenerate linear combining. The additional complication concerns taking into account that the elements~$\vv^1$ and~$\vv^2$ should form a basis of a Lie algebra, i.e., their commutator should lie in their span, $[\vv^1,\vv^2]\in\langle\vv^1,\vv^2\rangle$. Usually this places further restrictions on the admitted form of the elements.

To simply describe the conditions defining the different cases of the classification of two-dimensional subalgebras, we introduce the matrix notation 
\[
    A_{\mu_1\cdots\mu_n} := \left(\begin{array}{ccc} a^1_{\mu_1} & \cdots & a^1_{\mu_n} \\ a^2_{\mu_1} & \cdots & a^2_{\mu_n} \end{array}\right),
\]
where $\mu_i\in\{1,2,3\}$ and $n\leqslant3$. In what follows, the right hand side of a matrix equation $A_{\mu_1\cdots\mu_n}=0$ or a matrix inequality $A_{\mu_1\cdots\mu_n}\ne0$ is the zero matrix of the appropriate dimension. 

In the course of classification, we should investigate two principal cases.

1. $\mathrm{rank}(A_{123})=2$. 
This is the first cases which is partitioned into the three subcases 
\[
\mbox{(a)}\ \det A_{12}\ne0;\quad 
\mbox{(b)}\ \det A_{12}=0,\ \det A_{13}\ne0;\quad 
\mbox{(c)}\ \det A_{12}=0,\ \det A_{13}=0.
\]
In the last subcase we necessarily have $\det A_{23}\ne0$. 
By means of a change of the basis we at first set $A_{12}=E$, $A_{13}=E$ and $A_{23}=E$, respectively. 
Here $E$ is the $2\times 2$ identity matrix. 
If the new $h^2$ is nonvanishing, we set $h^2=-1$ using $\mathscr G_*(U)$ 
with the inverse~$U$ to a solution of the equation $\tilde U_u=-h^2(\tilde U)$.
In other words, up to $G^\sim$-equivalence we can always assume that $-h^2=\delta\in\{0,1\}$. 
We also set $h^1\in\{-1,0\}$ in a similar way if~$h^2=0$. 
Specifically, in subcase~(a) we further use the push-forward of $\mathscr T^x(a^1_3)$ to set $a^1_3=0$. 
As the resulting operators should commute, we derive that $a^2_3=0$ and $h^1_u=0$. 
This case hence leads to the first subalgebra from the list~\eqref{eq:OptimalListTwoDimension}.
In subcase~(b) we re-denote~$a^1_2$ by~$a$. Under the assumptions made, the commutator $[\vv^1,\vv^2]$ equals $-\vv_2$. 
Therefore, the condition $h_2=-1$ implies that $h^1_u=1$, i.e. we can set $h^1=u$ using a change of the basis and the push-forward of $\mathscr T^x(B_0)$
with certain~$B_0$.
This gives the second subalgebra from the list~\eqref{eq:OptimalListTwoDimension}.
If $h_2=0$, we obtain the third subalgebra. 
In subcase~(c), the corresponding subalgebra is commutative and hence $h^1_u=0$. 
Applying a scaling of~$\vv_2$ and the push-forward of~$\mathscr D^x(B_1)$ with certain~$B_1$, 
we simultaneously set $h^1,h^2\in\{-1,0\}$ and hence construct the fourth listed subalgebra. 

2. $\mathrm{rank}(A_{123})\leqslant1$. 
Up to a change of the basis, we can assume that $a^2_1=a^2_2=a^2_3=0$ and hence $h^2\ne0$, 
i.e., analogously to the previous case we can set $h^2=1$ by some $\mathscr G_*(U)$. 
Then up to a linear combining of~$\vv^1$ and~$\vv^2$ the commutation condition $[\vv^1,\vv^2]\in\langle\vv^1,\vv^2\rangle$ implies 
that $h^1_u=b=\const$ and, therefore, we can set $h^1=bu$. 
The four last algebras from the list~\eqref{eq:OptimalListTwoDimension} represent the subcases
\[
\mbox{(a)}\ a^1_1\ne0;\quad 
\mbox{(b)}\ a^1_1=0, a^1_2\ne0;\quad 
\mbox{(c)}\ a^1_1=0, a^1_2=0, a^1_3\ne0;\quad
\mbox{(d)}\ a^1_1=a^1_2=a^1_3=0,
\]
in which by a scaling of~$\vv_1$ we can set $a^1_1=1$, $a^1_2=1$, $a^1_3=1$ and $b=1$, respectively, 
For the basis elements to have the appropriate canonical form, we should additionally set
$a^1_3=0$ by some~$\mathscr T^x_*(B_0)$ and re-denote $a^1_2$ by~$a$ in subcase~(a) 
and also set $a^1_3\in\{-1,0\}$ by some~$\mathscr D^x_*(B_1)$ in subcase~(b)
and $b\in\{-1,0\}$ by some~$\mathscr D^x_*(B_1)$ and a scaling of~$\vv_1$ in subcase~(c).

This completes the proof of the theorem.
\end{proof}

Note that all except the last subalgebras from the lists~\eqref{eq:OptimalListOneDimension} and~\eqref{eq:OptimalListTwoDimension} represent parameterized classes of subalgebras rather than single subalgebras.

\section{Preliminary group classification}\label{sec:PreliminaryGroupClassification}

Based on Proposition~\ref{pro:OnPreliminaryGroupClassification1} and the above classification of subalgebras, 
we can obtain the extensions of the kernel algebra $\langle\p_t\rangle$ within the class~\eqref{eq:GenDiffEqs} 
by projections of inequivalent one- and two-dimensional subalgebras of the equivalence algebra~$\mathfrak g^\sim$ 
to the space of variables $(t,x,u)$. 
As a first step, for each of the subalgebras we solve the associated invariant surface condition for~$(f,g)$, namely, 
the system of equations $\xi f_x+\eta f_u=\varphi$, $\xi g_x+\eta g_u=\theta$, where the operator 
$\tau\p_t+\xi\p_x+\eta\p_u+\varphi\p_f+\theta\p_g$ runs through a basis of the subalgebra.

In Tables~\ref{tab:1DExtensions} and~\ref{tab:2DExtensions} we collect the general solutions of the invariant surface condition for~$(f,g)$ 
(or, in other words, the entire subclass of the corresponding invariant equations), 
which is associated with the one- and two-dimensional subalgebras of~$\mathfrak g^\sim$ listed in~\eqref{eq:OptimalListOneDimension} and~\eqref{eq:OptimalListTwoDimension}, respectively. 
In these tables, $\tilde f$ and $\tilde g$ are arbitrary functions of single arguments and  $c_1$ and $c_2$ are arbitrary constants such that 
$\tilde g\ne0$ and $c_2\ne0$.

\begin{table}[htb]\renewcommand{\arraystretch}{1.3}\belowcaptionskip=.7ex\abovecaptionskip=.0ex
\centering
\caption{One-dimensional Lie symmetry extensions for class~\eqref{eq:GenDiffEqs} related to~$\mathfrak g^\sim$.\label{tab:1DExtensions}}
\begin{tabular}{|c|l|l|l|}
 \hline
 N &\hfil $f$ &\hfil $g$ &\hfil Additional operator \\
 \hline
 1&$\tilde f(u+\delta\ln|x|)x^{2-a}$ & $\tilde g(u+\delta \ln|x|)x^{2-a}$ & $at\p_t+x\p_x-\delta\p_u$ \\
 2a&$\tilde f(u+\delta x)e^{-x} $ & $ \tilde g(u+\delta x)e^{-x} $ & $ t\p_t+\p_x-\delta\p_u$ \\
 2b&$\tilde f(x)e^u $ & $ \tilde g(x)e^u $ & $ t\p_t-\p_u$ \\
 3&$\tilde f(u+\delta x)$ & $\tilde g(u+\delta x)$ & $\p_x-\delta\p_u$ \\
 4&$\tilde f(x) $ & $\tilde g(x)$ &  $ \p_u $ \\
 \hline
\end{tabular}
\end{table}

\begin{table}[htb]\renewcommand{\arraystretch}{1.3}\belowcaptionskip=.7ex\abovecaptionskip=.0ex
\centering
\caption{Two-dimensional Lie symmetry extensions for class~\eqref{eq:GenDiffEqs} related to~$\mathfrak g^\sim$.\label{tab:2DExtensions}}
\begin{tabular}{|c|l|l|l|}
 \hline
 N &\hfil $f$ &\hfil $g$ &\hfil Additional operators \\
 \hline
 1&$ c_1e^ux^{2+\tilde\delta}$ & $ c_2e^ux^{2+\tilde\delta}$ & $x\p_x-\tilde\delta\p_u,\, t\p_t-\p_u$ \\
 2&$ c_1|u+x|^{1-a} $ & $c_2|u+x|^{2-a}$ & $ at\p_t+x\p_x+u\p_u,\,\p_x-\p_u $ \\
3a&$c_1e^{(2-a)u}$ & $c_2e^{(2-a)u}$ & $at\p_t+x\p_x-\p_u,\,\p_x$ \\
3b&$\tilde f(u)$ & $\tilde g(u)$ & $2t\p_t+x\p_x,\,\p_x$ \\
 4&$ c_1e^{u+\tilde\delta x}$ & $c_2e^{u+\tilde\delta x}$ & $t\p_t-\p_u,\, \p_x-\tilde\delta\p_u$ \\
 5&$c_1|x|^{2-a-b}$ & $c_2|x|^{2-a}$ & $at\p_t+x\p_x+bu\p_u,\, \p_u$ \\
 6&$c_1e^{(1+b)x}$ & $c_2e^x$ & $t\p_t-\p_x+bu\p_u,\, \p_u$ \\
 7&$c_1e^{\delta x}$ & $c_2$ & $\p_x-\delta u\p_u,\,\p_u$ \\
 8&$0 $ & $\tilde g(x) $ & $\p_u,\, u\p_u$ \\
 \hline
\end{tabular}
\end{table}

The second algebra from the list of one-dimensional subalgebras~\eqref{eq:OptimalListOneDimension} is associated with a symmetry extension of an equation from the class~\eqref{eq:GenDiffEqs} if and only if at least one of its parameters~$\delta$ and~$\tilde\delta$ does not vanish. In addition, to find the corresponding ansatzes for~$f$ and~$g$ it is necessary to consider different cases of values of the parameters. This is why this subalgebra leads to two cases (2a and 2b) of Table~\ref{tab:1DExtensions}. 
Analogously, equations from the class~\eqref{eq:GenDiffEqs} are invariant with respect to the projections of the first, fourth or sixth algebras from the list of two-dimensional subalgebras~\eqref{eq:OptimalListTwoDimension} if  and only if $\delta\ne0$, i.e., we can assume that $\delta=1$. 
For the third algebra we should have either $\delta\ne0$ (then we can again assume that $\delta=1$) or $(\delta,a)=(0,2)$ that gives Cases~3a and~3b, respectively. 

There are several reasons why Tables~\ref{tab:1DExtensions} and~\ref{tab:2DExtensions} do not give a proper classification result. 
We present these reasons in the form of the following series of remarks. 

\begin{remark}\label{rem:OnEqiuvalenceOfDiffCasesIn2DExtensionsForGenDiffEqs}
As the whole consideration is done up to $G^\sim$-equivalence, we should additionally factorize 
the general solutions of the invariant surface conditions for~$f$ and~$g$ with respect to this equivalence. 
Using transformations from~$G^\sim$, in Table~\ref{tab:2DExtensions} we can set $c_2=1$ (by scaling of~$t$ and alternating its sign) 
and, in Cases~$5_{b=0}$, $6_{b=0}$ and~$7_{\delta=0}$, $c_1=0$ 
(by the transformation~\eqref{eq:TransForGenDiffEqsWithFEqualG} with $c=c_1/c_2$, cf.\ Remark~\ref{rem:OnParticularSiplificationOfGenDiffEqs}).
For the other values of the parameters~$b$ and~$\delta$ in these cases, the constant~$c_1$ can be assumed, up to $G^\sim$-equivalence, to belong to~$\{0,1\}$. If~$a\ne2$ in Case~3, we can scale the value~$2-a$ to~1.
\end{remark}

\begin{remark}\label{rem:OnNonmaximalityOf2DExtensionsForGenDiffEqs1}
Extensions presented in Tables~\ref{tab:1DExtensions} and~\ref{tab:2DExtensions} are not necessarily maximal even for the general values of the parameter-functions~$\tilde f$ and~$\tilde g$ or the constant parameters~$c_1$ and~$c_2$. It lies in the nature of preliminary group classification that equations can admit operators which are not projections of operators of the equivalence algebra. For example, in the last case of Table~\ref{tab:2DExtensions} any corresponding equation is linear and therefore admits an infinite-dimensional Lie invariance algebra including also the operators of the form $\varphi(t,x)\p_u$, where $\varphi$ runs through the set of solutions of the equation under consideration. (Of course, for certain values of~$g$ this equation possesses an even wider Lie invariance algebra, cf.~\cite{lie81Ay,ovsi82Ay}.) 
A similar remark is true for Case~$5_{b=0}$, (resp.\ Case~$6_{b=0}$, resp.\ Case~$7_{\delta=0}$) of Table~\ref{tab:2DExtensions} since each of the equations corresponding to this case is reduced by an equivalence transformation to the linear equation with $f=0$ and $g=|x|^{2-a}$ (resp.\ $g=e^x$, resp.\ $g=1$), cf.\ Remark~\ref{rem:OnEqiuvalenceOfDiffCasesIn2DExtensionsForGenDiffEqs}.
\end{remark}

\begin{remark}\label{rem:OnNonmaximalityOf2DExtensionsForGenDiffEqs2}
What is more essential is that presented extensions are not maximal even among extensions related to subalgebras of~$\mathfrak g^\sim$. 
In particular, Case~$3_{\delta=0}$ of Table~\ref{tab:1DExtensions} coincides by the arbitrary elements with Case~3b of Table~\ref{tab:2DExtensions} and hence should be excluded from the extension list.
Within Table~\ref{tab:2DExtensions}, if $a=2$ the arbitrary elements in Cases~3a and~$5_{b=0}$ coincide with those of Case~$7_{\delta=0}$. Hence in Case~$7_{\delta=0}$ we have the additional operator~$2t\p_t+x\p_x$ induced by the operator $\DDD^x+2\DDD^t$. The algebra presented in Case~$3_{a\ne2}$ is also not maximal, cf.~Case~1 of Table~\ref{tab:34DExtensions}. Cases~4,~6 and~7 admit additional extensions by the operator~$u\p_u$ if $c_1=0$ or $e^{-c_1u/c_2}\p_u$ if $c_1\ne0$ and $b=0$ (resp.\ $\delta=0$), 
owing to the connection of these cases with Case~8 via the transformation~\eqref{eq:TransForGenDiffEqsWithFEqualG}. 
\end{remark}

\begin{remark}
An effect of the lack of maximality of extensions is that under the simplification of the form of arbitrary elements by equivalence transformations the corresponding invariance algebra may be replaced a similar one. Thus, under setting~$c_1=0$ in Cases~$5_{b=0}$, $6_{b=0}$ and~$7_{\delta=0}$ the basis element~$\p_u$ is replaced by $u\p_u$.
\end{remark}

\looseness=-1
In order to complete the preliminary group classification of the class~\eqref{eq:GenDiffEqs}, 
we should at first construct the exhaustive list of $G^\sim$-inequivalent subalgebras of~$\mathfrak g^\sim$ 
whose projections to the space of the variables~$(t,x,u)$ are 
Lie invariance algebras of equations from the class~\eqref{eq:GenDiffEqs}. 
For convenience such subalgebras will be called appropriate. 
Then we should study the problem whether these subalgebras are maximal among 
the subalgebras with the same property for a certain subclass of  the class~\eqref{eq:GenDiffEqs}. 
The majority of one- and two-dimensional subalgebras of~$\mathfrak g^\sim$ are appropriate. 
This is why for subalgebra dimensions one and two it is not too important whether all or only appropriate subalgebras are classified 
but this is not the case for greater dimensions. 
As the arbitrary elements~$f$ and~$g$ depend on two arguments, the condition that the associated projection is a Lie invariance algebra of an equation from the class~\eqref{eq:GenDiffEqs} is a strong restriction for subalgebras of~$\mathfrak g^\sim$ of dimension greater than two and even leads to the boundedness of dimension of such subalgebras.

Let $\mathfrak g^\sim_1=\langle\mathcal \DDD^t,\GG(h)\rangle$, where $h$ runs through the set of smooth functions of~$u$.
For a subalgebra~$\mathfrak s$ of~$\mathfrak g^\sim$, we denote $\dim\mathfrak s\cap\mathfrak g^\sim_1$ by~$m_{\mathfrak s}$.

\begin{lemma}\label{lem:OnAppropriateAlgebras}
$\DDD^t\not\in\mathfrak s$ and $m_{\mathfrak s}\leqslant2$ for any appropriate subalgebra~$\mathfrak s$ of~$\mathfrak g^\sim$.
\end{lemma}

\begin{proof}
Let~$\mathfrak s$ be an appropriate subalgebra of~$\mathfrak g^\sim$. 
Then the system of invariant surface conditions associated with elements of~$\mathfrak s$ should have a solution~$(f^0,g^0)$ with $g^0\ne0$. 
The invariant surface condition for~$g$ associated with the operator~$\DDD^t$ is $g=0$ that contradicts the auxiliary inequality $g\ne0$. 
Hence $\DDD^t\not\in\mathfrak s$. 

In what follows, the indices~$i$ and~$j$ run from~1 to~3.
Suppose that the subalgebra~$\mathfrak s$ contains at least three linearly independent elements from~$\mathfrak g^\sim_1$, 
$\vv^i=\GG(h^i)+a^i\DDD^t$. The corresponding invariant surface conditions for~$g$ form the system $h^ig_u+a^ig=0$. 
We consider it as a homogenous system of linear algebraic equations with respect to $(g_u,g)$. This system should have a nonzero solution since $g\ne0$. 
Therefore $h^ia^j-h^ja^i=0$. 
In view of the linear independence of~$\vv^1$, $\vv^2$ and~$\vv^3$, this implies that all $a^i=0$ and thus $g_u=0$. 
Now we interpret the system of invariant surface conditions $h^if_u+h^i_uf+h^i_{uu}g=0$ for~$f$ 
as a homogenous system of linear algebraic equations with respect to $(f_u,f,g)$. 
As $g\ne0$, this system should possess a nonzero solution and hence the determinant of its matrix vanishes. 
At the same time, the determinant coincides with the Wronskian of the linearly independent functions~$h^1$, $h^2$ and~$h^3$, which is not equal to zero. 
The contradiction obtained implies that $m_{\mathfrak s}\leqslant2$.
\end{proof}

\begin{corollary}\label{cor:OnAppropriateAlgebras1}
Any appropriate subalgebra~$\mathfrak s$ of~$\mathfrak g^\sim$ is of dimension not greater than four. 
\end{corollary}

\begin{proof}
The projection of any element from $\mathfrak s\setminus\mathfrak g^\sim_1$ to $\langle\DDD^x,\p_x\rangle$ should be nonzero. 
Therefore, $\dim\mathfrak s\leqslant\dim\langle\DDD^x,\p_x\rangle+m_{\mathfrak s}=4$.
\end{proof}

\begin{corollary}\label{cor:OnAppropriateAlgebras2}
$\mathfrak s\cap\mathfrak g^\sim_1=\mathfrak s\cap\langle\GG(h)\rangle$ for any appropriate subalgebra~$\mathfrak s$ of~$\mathfrak g^\sim$ with~$m_{\mathfrak s}=2$, where $h$ runs through the set of smooth functions of~$u$. 
\end{corollary}

\begin{proof}
As $m_{\mathfrak s}=2$, the subalgebra~$\mathfrak s$ contains two linearly independent elements from~$\mathfrak g^\sim_1$, 
$\vv^i=\GG(h^i)+a^i\DDD^t$, $i=1,2$. 
Analogously to the proof of Lemma~\ref{lem:OnAppropriateAlgebras}, we consider the system of the invariant surface conditions $h^ig_u+a^ig=0$ for~$g$ associated with~$\vv^i$ as a homogenous system of linear algebraic equations with respect to $(g_u,g)$, which has a nonzero solution since $g\ne0$. Therefore, the determinant of its matrix equal zero, $h^1a^2-h^2a^1=0$. 
In view of the linear independence of~$\vv^1$ and~$\vv^2$, this implies that $a^1=a^2=0$. 
\end{proof}

As we have classified all one- and two-dimensional subalgebras of~$\mathfrak g^\sim$, it is enough to describe appropriate subalgebras only of dimensions greater than~2.

\begin{theorem}\label{thm:InequivalentThreeFourDimensionalSubalgebras}
A complete list of $G^\sim$-inequivalent appropriate subalgebras of~$\mathfrak g^\sim$ is exhausted by the following subalgebras:
\begin{gather*}
    \langle\DDD^x+\GG(2),\p_x,\DDD^t-\GG(1)\rangle, \quad
    \langle\DDD^x+2\DDD^t+b\GG(u),\p_x,\GG(1)\rangle, \\
    \langle\DDD^x+a\DDD^t,\GG(1),\GG(u)\rangle, \quad 
    \langle\p_x-\delta\DDD^t,\GG(1),\GG(u)\rangle, \quad 
    \langle\DDD^x+2\DDD^t,\p_x,\GG(1),\GG(u)\rangle,
\end{gather*}
where $a$, $b$ and $\delta$ are constants and we can assume that $\delta\in\{0,1\}$.
\end{theorem}

\begin{proof}
Let~$\mathfrak s$ be an appropriate subalgebra of~$\mathfrak g^\sim$ and $\dim\mathfrak s\geqslant3$. Then, $m_{\mathfrak s}>0$. Hence we should consider only the cases $m_{\mathfrak s}=1$ and $m_{\mathfrak s}=2$. 

The condition $m_{\mathfrak s}=1$ means that the subalgebra $\mathfrak s$ contains exactly one operator of the form $\vv^1=\GG(h^1)+a^1_2\DDD^t$, where $h^1\ne0$, in view of Lemma~\ref{lem:OnAppropriateAlgebras}. By scaling of~$\vv^1$ we can set $a^1_2=-1$ if $a^1_2\ne0$. Moreover, as the function-parameter~$h^1$ does not vanish it can be set to~$1$ upon using $\mathscr G_*(U)$ with the inverse~$U$ to a solution $\tilde U=\tilde U(u)$ of the equation $\tilde U_u=h^1(\tilde U)$. As a result we have two $G^\sim$-inequivalent forms for $\vv^1$: (i) $\vv^1=\GG(1)-\DDD^t$, (ii) $\vv^1=\GG(1)$. The conditions $\dim\mathfrak s\geqslant3$ and $m_{\mathfrak s}=1$ simultaneously imply that $\dim\mathfrak s=3$. This is why we should have two more linearly independent operators of the form $\vv^i=a^i_1\DDD^x+a^i_2\DDD^t+a^i_3\p_x + \GG(h^i)$, $i=2,3,$ from $\mathfrak s\setminus\mathfrak g^\sim_1$ for which $\rank(a^i_1,a^i_3)_{i=2,3}=2$, cf.\ the proof of Corollary~\ref{cor:OnAppropriateAlgebras1}. By linear combining of~$\vv^2$ and~$\vv^3$ we set $a^2_1=a^3_3=1$ and $a^2_3=a^3_1=0$. 

In subcase~(i) we additionally subtract $a^i_2\vv^1$ from~$\vv^i$ to obtain $a^i_2=0$ in the new operator~$\vv^i$, $i=2,3$. The simplified form of $\vv^2$ and $\vv^3$ is $\vv^2=\DDD^x+\GG(h^2)$ and $\vv^3=\DDD^t+\GG(h^3)$, respectively. As $\mathfrak s$ is a Lie algebra and $\{\vv^1,\vv^2,\vv^3\}$ is a basis of $\mathfrak s$, any commutator of $\vv$'s should lie in their linear span. This in particular implies that the operators $\vv^2$ and $\vv^3$ should commute with~$\vv^1$, which is equivalent to the conditions $h^2_u=0$ and $h^3_u=0$. Then, the commutator $[\vv^2,\vv^3]$ equals $\p_x$, which should belong to $\mathfrak s$. Therefore, $h^3=0$. The complete system of invariant surface conditions associated with $\{\vv^1,\vv^2,\vv^3\}$ has a solution with nonvanishing $g$ if and only if $h^2=2$. As a result we obtain the first listed subalgebra. 

Analogously, in subcase~(ii) we have $[\vv^1,\vv^i]=\GG(h^i_u)=b^i\vv^1$, where $b^i=\const$, $i=2,3$, i.e., up to linear combining of $\vv^i$ with $\vv^1$, $h^i=b^iu$. The condition $[\vv^2,\vv^3]=\p_x\in\mathfrak s$ yields that $a^3_2=b^3=0$. In order to provide the requested compatibility of the entire system of the associated invariant surface conditions with the inequality $g\ne0$, we necessarily have $a^2_2=2$. Re-denoting $b^2=b$, we recover the second subalgebra from the above~list. 

If $m_{\mathfrak s}=2$, the subalgebra~$\mathfrak s$ contains two linearly independent operators $\vv^i=\GG(h^i)$, $i=1,2$. Similarly to case 2d of the proof of Theorem~\ref{thm:InequivalentTwoDimensionalSubalgebras}, we can assume up to $G^\sim$-equivalence and a change of the basis in $\langle\vv^1,\vv^2\rangle$ that $h^1=1$ and $h^2=u$. Consider any $\vv=a_1\DDD^x+a_2\DDD^t+a_3\p_x+\GG(h)$ from the complement to $\langle\GG(1),\GG(u)\rangle$ in~$\mathfrak s$. Lemma~\ref{lem:OnAppropriateAlgebras} implies that $(a_1,a_3)\ne(0,0)$. Therefore, $[\vv^i,\vv]\in\langle\vv^1,\vv^2\rangle$. The last condition is equivalent to $h_u,uh_u{-}h\in\langle1,u\rangle$. Consequently, we obtain $h\in\langle1,u\rangle$. Hence, up to linear combining with elements from $\langle\GG(1),\GG(u)\rangle$, we can always assume that $h=0$. In other words, the subalgebra $\mathfrak s$ can be represented as a direct sum of the algebra $\langle\GG(1),\GG(u)\rangle$ and a subalgebra of $\mathfrak g^\sim_2=\langle\DDD_x,\DDD_t,\p_x\rangle$. $G^\sim$-inequivalent subalgebras of~$\mathfrak g^\sim_2$ that do not contain the operator~$\DDD^t$ are exhausted by the algebras $\langle\DDD^x+a\DDD^t\rangle$, $\langle\p_x-\delta\DDD^t\rangle$ and $\langle\DDD^x+a\DDD^t,\p_x\rangle$, cf.\ the proofs of Theorems~\ref{thm:InequivalentOneDimensionalSubalgebras} and~\ref{thm:InequivalentTwoDimensionalSubalgebras}. In the last subalgebra, owing to the required compatibility of the system of invariant surface conditions associated with~$\mathfrak s$ we have~\mbox{$a=2$}. 

This completes the proof of the theorem.
\end{proof}

\looseness=-1
The symmetry extensions induced by subalgebras from Theorem~\ref{thm:InequivalentThreeFourDimensionalSubalgebras} are collected in Table~\ref{tab:34DExtensions}, where $a$ is an arbitrary constant, $a\ne2$. 
Note that the extension induced by the second subalgebra is not maximal among extensions related to~$\mathfrak g^\sim$. This is why we do not include it into Table~\ref{tab:34DExtensions}. The general solution of the associated system of invariant surface conditions is $f=c_1$ and $g=c_2$, where~$c_1$ and~$c_2$ are arbitrary constants, $c_2\ne0$. Such values of arbitrary elements correspond to the potential Burgers equation or the linear heat equation if $c_1\ne0$ or $c_1=0$, respectively. 
The linear heat equation is given by Case~4 of Table~\ref{tab:34DExtensions} and the potential Burgers equation, which additionally possesses the Lie symmetry operator $e^{-c_1u/c_2}\p_u$ induced by $\GG(e^{-c_1u/c_2})$, is reduced to the same case by a transformation similar to~\eqref{eq:TransForGenDiffEqsWithFEqualG}, 
cf.\ Remark~\ref{rem:OnNonmaximalityOf2DExtensionsForGenDiffEqs2}. 
Analogously, we should choose $\delta=1$ in the fourth subalgebra for the associated extension to be maximal.  

\begin{table}[htb]\renewcommand{\arraystretch}{1.4}\abovecaptionskip=0ex\belowcaptionskip=.7ex
\centering
\caption{Lie symmetry extensions for class~\eqref{eq:GenDiffEqs} related to~$\mathfrak g^\sim$ of dimension greater than two.\label{tab:34DExtensions}}
\begin{tabular}{|c|l|l|l|}
 \hline
 N &\hfil $f$ &\hfil $g$ &\hfil Additional operators \\
 \hline
 1&$ce^u$ & $e^u$ & $x\p_x+2\p_u,\,\p_x,\,t\p_t-\p_u$ \\
 2&$0$ & $|x|^{2-a}$ & $at\p_t+x\p_x,\,\p_u,\,u\p_u$ \\
 3&$0$ & $e^x$ & $t\p_t-\p_x,\,\p_u,\,u\p_u$ \\
 4&$0$ & $1$ & $2t\p_t+x\p_x,\,\p_x,\,\p_u,\,u\p_u$ \\
 \hline
\end{tabular}
\end{table}

Summing up the whole consideration of the present paper, we prove the following theorem:

\begin{theorem}\looseness=-1
The complete preliminary group classification of class~\eqref{eq:GenDiffEqs} is split into Tables~\ref{tab:1DExtensions}--\ref{tab:34DExtensions}, 
where $\delta\ne0$ in Case~3 of Table~\ref{tab:1DExtensions} and 
in Table~\ref{tab:2DExtensions} we should globally set $c_2=1$, exclude Case~3a and assume that $\tilde\delta\ne-2$ in Case~1, $\tilde\delta\ne0$ in Case~4, $b\ne0$ in Cases~5 and~6, and $\delta=1$ in Case~7.
\end{theorem}

\begin{remark}\label{rem:OnNadj10Ay}
Table~3 from~\cite{nadj10Ay}, summing up the partial preliminary group classification of the class~\eqref{eq:GenDiffEqs} therein, is incorrect. Neither are all of the equations listed really invariant under the operators presented in the table, nor are these operators proper additional operators in view of the kernel $\langle\p_t\rangle$. The main problem is that the basis element~$\p_t$ of the kernel is involved by linear combining to these additional operators which, moreover, are not linearly independent. 
\end{remark}

The number of inequivalent cases to be investigated under the usage of the entire infinite-dimensional equivalence algebra~$\mathfrak g^\sim$ is rather small. This is due to the greater effectiveness of the adjoint action of the whole equivalence group, which allows for stronger simplifications under classification of inequivalent subalgebras. By using only a finite-dimensional subalgebra of~$\mathfrak g^\sim$ as usually done, the number of cases of extensions to be treated is generally greater. This is one more justification why it is favorable to use complete preliminary group classification rather than partial preliminary group classification.

\section{Conclusion}\label{sec:Conclusion}

The main aim of this paper is a careful explanation of the technique of preliminary group classification, its status in the picture of group classification, its benefits and its limitations. These points are those we consider to be mainly lacking so far. While preliminary group classification is generally attractive due to the relative simplicity of its algorithm, various of the results obtained by now using this approach have only little practical relevance, since they are presented without a detailed analysis of the class of differential equations. Moreover, in various papers only partial preliminary group classification was carried out, without indicating a sound physical justification for the chosen subalgebras of the respective equivalence algebras. Indeed, in some instances this choice might been motivated for the sake of pure mathematical convenience, which counteracts the initial aim of group classification of differential equations.

\looseness=-1
In the present paper we substantially enhance the existing framework of preliminary group classification. 
We show that it is possible and convenient to treat subalgebras of the entire equivalence algebra even in the case if this algebra is infinite dimensional. This is the principal difference compared to existing works on the subject of preliminary group classification, in which the problem is only partially solved by involving classification of subalgebras of a fixed finite-dimensional subalgebra of the equivalence algebra with respect to restricted adjoint actions. Furthermore, it is emphasized that only appropriate subalgebras satisfying certain properties should be classified. 

The algorithm of \emph{complete preliminary group classification} can be summed up as follows:
\begin{itemize}\itemsep=0ex
\item 
Find the equivalence algebra~$\mathfrak g^\sim$ and the equivalence group~$G^\sim$ of the class~$\mathcal L|_{\mathcal S}$ under consideration. 
\item
Classify appropriate subalgebras of~$\mathfrak g^\sim$ up to~$G^\sim$-equivalence, each of which satisfies the properties below:
\begin{itemize}
\item
It contains the kernel algebra~$\mathfrak g^\cap$ of~$\mathcal L|_{\mathcal S}$.
\item
The associated system of invariant surface conditions with respect to the arbitrary elements is compatible.
\item
It is the maximal subalgebra among all subalgebras of~$\mathfrak g^\sim$ that have the same set of solutions for the associated systems of invariant surface conditions.
\end{itemize}
\item
For each of the listed subalgebras, find the general solution of the associated system of invariant surface conditions with respect to the arbitrary elements.
\item
Simplify these solutions using transformations from~$G^\sim$ whose push-forwards to vector fields preserve the corresponding subalgebras of~$\mathfrak g^\sim$, i.e., these transformations lie in the normalizers of the corresponding subgroups of~$G^\sim$.
\end{itemize}

The systematic approach of complete preliminary group classification is exemplified with the class of generalized diffusion equation~\eqref{eq:GenDiffEqs} that was recently attempted to be investigated in~\cite{nadj10Ay} using symmetry tools. Owing to the number of inconveniences of~\cite{nadj10Ay}, we regard the class~\eqref{eq:GenDiffEqs} as well-suited to explain the methodology of preliminary group classification. We use both the framework of the infinitesimal and the direct methods to derive the equivalence algebra and the equivalence group of the class~\eqref{eq:GenDiffEqs}. In addition, the direct method also allows us to obtain the classifying equations of admissible transformations. Similarly to the determining equations of Lie symmetries, these classifying equations are too difficult to be solved directly, which at once limits the chance to obtain a complete group classification of the class~\eqref{eq:GenDiffEqs}.

It is important to indicate once more that the extensions of the kernel algebra constructed in this paper by using preliminary group classification are not necessarily maximal. That is, there are various equations in the class~\eqref{eq:GenDiffEqs} which have the maximal Lie invariance algebras wider than the associated subalgebras of the equivalence algebra. This observation is another way of proving that the class~\eqref{eq:GenDiffEqs} is not normalized.

\section*{Acknowledgements}

We are thankful to the referees for helpful suggestions that have led to improvements of the paper.
The research of EDSCB and ROP was supported by the Austrian Science Fund (FWF), project P20632. AB is a recipient of a DOC-fellowship of the Austrian Academy of Sciences.

{\small\setlength{\itemsep}{0ex}

}

\end{document}